\documentclass[12pt]{article}
\usepackage[english]{babel}
\usepackage[utf8]{inputenc}
\usepackage[T1]{fontenc}
\usepackage{amsbsy}
\usepackage{amsfonts}
\usepackage{amsmath}
\usepackage{amssymb}
\usepackage{amsthm}
\usepackage{graphicx} 
\usepackage[pdftex]{color}
\usepackage{enumerate}
\usepackage{float}
\usepackage{geometry, calc, color}
\usepackage{indentfirst}
\usepackage{longtable}
\usepackage{multirow}
\usepackage{natbib}
\usepackage{xr-hyper}
\usepackage{hyperref}
\usepackage[dvipsnames]{xcolor}
 \usepackage{booktabs}
 \usepackage{geometry}
 \usepackage{caption}
 \usepackage{subcaption}

\usepackage{bm}
\usepackage{booktabs}
\usepackage{xcolor}
\usepackage{hyperref}
\usepackage{url}
\usepackage{multicol}
\usepackage{rotating}

\newtheorem{theorem}{Theorem}[section]
\newtheorem{definition}{Definition}[section]

\newtheorem{remark}[theorem]{Remark}
\newtheorem{assumption}[theorem]{Assumption}

\hypersetup{
  colorlinks=true,
  linkcolor=red,
  citecolor=blue,
  filecolor=red,
  urlcolor=RoyalBlue,
}
\newcommand{\blind}{1}

\protect

\begin{document}

\def\spacingset#1{\renewcommand{\baselinestretch}%
{#1}\small\normalsize} \spacingset{3}

\spacingset{1.45} 

\if1\blind
{
  \title{\bf Flexible Bivariate INGARCH Process With a Broad Range of Contemporaneous Correlation}
  \author{Luiza S.C. Piancastelli$^\sharp$\footnote{Email: \texttt{luiza.piancastelli@ucdconnect.ie}},\,\, Wagner Barreto-Souza$^\star$\footnote{Corresponding author. Email: \texttt{wagner.barretosouza@kaust.edu.sa}}\, and Hernando Ombao$^\star$\footnote{Email: \texttt{hernando.ombao@kaust.edu.sa}}\hspace{.2cm}\\
  	{\normalsize \it $^\sharp$School of Mathematics and Statistics, University College Dublin, Dublin, Ireland}\\
    {\normalsize \it $^\star$Statistics Program, King Abdullah University of Science and Technology, Thuwal, Saudi Arabia}}
  \maketitle
} \fi
\if0\blind
{
  \bigskip
  \bigskip
  \bigskip
  \begin{center}
    {\LARGE\bf Flexible Bivariate INGARCH Process With a Broad Range of Contemporaneous Correlation}
\end{center}
\medskip
} \fi

\bigskip
\addtocontents{toc}{\protect\setcounter{tocdepth}{1}}

\begin{abstract}
We propose a novel flexible bivariate conditional Poisson (BCP) INteger-valued Generalized AutoRegressive Conditional Heteroscedastic (INGARCH) model for correlated count time series data. Our proposed BCP-INGARCH model is mathematically tractable and has as the main advantage over existing bivariate INGARCH models its ability to capture a broad range (both negative and positive) of contemporaneous cross-correlation which is a non-trivial advancement. Properties of stationarity and ergodicity for the BCP-INGARCH process are developed.  Estimation of the parameters is performed through conditional maximum likelihood (CML) and finite sample behavior of the estimators are investigated through simulation studies. Asymptotic properties of the  CML estimators are derived. Additional simulation studies compare and contrast methods of obtaining standard errors of the parameter estimates, where a bootstrap option is demonstrated to be advantageous. Hypothesis testing methods for the presence of contemporaneous correlation between the time series are presented and evaluated. We apply our methodology to monthly counts of hepatitis cases at two nearby Brazilian cities, which are highly cross-correlated. The data analysis demonstrates the importance of considering a bivariate model allowing for a wide range of contemporaneous correlation in real-life applications.
\end{abstract}
{\it \textbf{Keywords}:} Asymptotics; Cross-correlation; Ergodicity; Stability theory; Multivariate count time series.


\section{Introduction}\label{intro}

There is an increasing interest in count time series models because of their applications especially during a pandemic where it is crucial to predict counts of cases or hospitalizations. Count time series data is collected and studied in many fields including public health, cybersecurity, criminology, and medical science. In this paper, our ultimate goal is to study the dynamics in hepatitis counts between cities in Brazil and to address forecast procedures. Univariate count time series models based on hidden Markov chains, INAR (INteger-valued AutoRegressive), and INGARCH (INteger-valued Generalized AutoRegressive Conditional Heteroskedastic) approaches, among others, were introduced and explored in several papers. For a general account on univariate count models, please see \cite{kedfok2002}, \cite{fok2011}, \cite{davetal2015}, and \cite{wei2018}. 

Our goal here is to develop a model for multivariate count time series that can capture cross-dependencies between the different components. Recent works have emerged on this topic, but several limitations have yet to be addressed. As noted by \cite{camtri1998}, \cite{junetal2011}, and \cite{kar2016}, models for multivariate count time series are rather sparse mainly due to the analytical and computational challenges. 

One popular approach to modeling multivariate count time series uses the idea of thinning operators \citep{stevan1979} that were introduced to statistics by \cite{lat1997}, \cite{pedkar2011}, \cite{karped2013}, \cite{pedkar2013a}, \cite{pedkar2013b}, and \cite{scoetal2014}. A limitation of these thinning-based models is that in general they are not able to model negative cross-correlation and the associated likelihood function is cumbersome. \cite{livetal2018} introduced a bivariate class of thinning-based count processes that accommodates both positive and negative cross-correlation, overcoming one of the main limitations of the previous models. However, the proposed model has Poisson distributed marginals and therefore is not adequate to deal with overdispersion, a phenomenon that occurs frequently in practice. In addition to that, it is argued by \cite{daretal2019} that estimation and forecasting are cumbersome under this model. In this last paper, the authors propose a bivariate thinning-based model with a flexible autocorrelation structure. Although point predictors have been discussed, there is a lack of more formal approaches to perform inference, where the same criticism on the cumbersome likelihood approach to the model by \cite{livetal2018} applies. Another challenge of such models is that is not clear how to generalize them to dimensions greater than 2. 

Another approach to analyzing multivariate count time series are the latent factor-based models. In a pioneering work, \cite{joretal1999} proposed a state-space multivariate count time series model, where the counts are assumed to be conditionally independent according to Poisson distributions given a common latent gamma process. This model is nicely motivated in a context involving counts of emergency room visits for 4 respiratory diseases, where the common underlying tendency for pollution to cause respiratory diseases is interpreted through the latent gamma process. Although this assumption makes sense in this problem, it is hard to justify only one common latent process in general practical situations. In this direction, \cite{junetal2011} introduced a multivariate count time series model based on latent factor processes, so obtaining a more flexible correlation structure for the multivariate counts than the model by \cite{joretal1999}. The authors assumed that the latent factors are independent first-order autoregressive Gaussian processes. 

Recently, \cite{wanwan2018} proposed a non-stationary multivariate count time series model driven by latent factors, where the counts are assumed to be conditionally independent Poisson distributed given the factors, as considered in \cite{joretal1999} and \cite{junetal2011}. The difference is how the factors are considered in \cite{wanwan2018}. More specifically, the factors are expressed as a linear combination of possibly low dimensional factors. Since the dynamic of the latent processes is not specified in that paper, in contrast as in \cite{joretal1999} and \cite{junetal2011}, the marginal distribution of the counts cannot be explicit. Further, the likelihood function is not available in an analytic form, which implies difficulties in performing model inference. To overcome this problem, the authors proposed a two-step estimation procedure by combining a pseudo-maximum likelihood approach by \cite{gouetal1984} to obtain estimates for the regression coefficients, and then calibration of the factor loading matrix is performed through eigenanalysis in the second step. A major shortcoming of this approach is that it is extremely complicated to perform prediction. Thus, this model cannot be readily applied to address public health needs to conduct forecasting. 

It is important to note that the correlation of the last three aforementioned papers is driven totally through the latent factors. In other words, the cross-covariance of contemporaneous counts equals the covariance of the associated contemporaneous factors; for instance, see Eq. (2) from \cite{wanwan2018}. This phenomenon happens because of the conditional independence of the counts given the factors, that is all source of correlation is fully captured only by the latent processes. That is, conditional on the latent sources, there is no additional source of cross-dependence between components.

Proposed by \cite{hei2003}, \cite{ferland2006} and \cite{fok2009}, the INteger-valued Generalized Autoregressive Conditional Heterocedastic (INGARCH) models are a popular and tractable alternative to model count time series. These models can be seen as an integer-valued counterpart of the GARCH models by \cite{boll1986}. A univariate Poisson INGARCH(1,1) model specifies that the time series of counts $\{Y_t\}_{t\in\mathbb Z}$ is defined by 
$Y_t|\mathcal F_{t-1} \sim \mbox{Poisson}({\lambda_t})$, $\lambda_t = \omega + \alpha \lambda_{t-1} + \beta Y_{t-1}$, 
where $\mathcal F_{t-1}$ is the $\sigma$-field generated by $\{Y_{t-1},Y_{t-2},\ldots\}$, $\omega>0$, $\alpha\geq0$, $\beta\geq0$. Alternatives to the Poisson assumption motivated several works in the literature, as the negative binomial \citep{zhu2010,chrfok2014}, infinitely divisible \cite{goncalves2015}, exponential family \citep{davis2016}, and mixed Poisson \citep{chrfok2015,bs2019} INGARCH models, to name a few. Further, the linearity assumption was relaxed in \cite{fok2011} and \cite{fok2012}, which introduced the log-linear and nonlinear INGARCH models.  

Although there is abundant work on univariate INGARCH models, multivariate extensions are still scarce in the field. The first work dealing with this topic is due to \cite{liu2012}, where a bivariate Poisson INGARCH model was proposed. This bivariate process was also studied by \cite{leeetal2018}, where asymptotic properties of estimators and a parameter-change test were addressed. Recently, \cite{cuizhu2018} introduced another Poisson bivariate count time series model allowing for both negative and positive contemporaneous correlation (also known as cross-correlation). A drawback of the model by \cite{liu2012} and \cite{leeetal2018} is that negative cross-correlation is not allowed in contrast with the model by \cite{cuizhu2018}. On the other hand, in the former, the supported range of cross-correlation is very limited. In these papers, we have two major problems regarding that range: (i) natural restriction of the parameter space due to the baseline bivariate discrete distribution; (ii) the parameter space (related to the cross-correlation parameter) of the baseline count distribution depends on the marginal means. These points imply that there is a severe limitation in the correlation between the count time series which these models can capture, limiting their practical applicability. Our proposed model deals with both issues since the parameter space of its correlation parameter is $\mathbb R$-valued (so it does not depend on the means) and our baseline distribution allows for a broad range of correlation. In Subsection \ref{other_ingarch}, we provide a detailed discussion on these restrictions and how our model overcomes them.

We propose a novel bivariate INteger-valued Generalized AutoRegressive Conditional Heteroscedastic (INGARCH) model for the statistical analysis of correlated count time series data. More specifically, we introduce and study a new flexible bivariate conditional Poisson (BCP) INGARCH model, which is mathematically tractable and whose main advantage over the existing bivariate count time series models, by \cite{liu2012}, \cite{leeetal2018}, and \cite{cuizhu2018}, is its ability to capture a broad range of both negative and positive contemporaneous correlation. Along with the paper, we argue that such a broad range is very important to model properly high correlated count time series. Besides, we derive the theoretical properties of the BCP-INGARCH model such as conditions to ensure stationarity and ergodicity, and asymptotics on the conditional maximum likelihood estimators, as well as a full discussion on the statistical modeling including point estimation, procedures to obtain standard errors, hypothesis testing on the presence of cross-correlation, and forecasting.

It is worth to mention two other related works about multivariate INGARCH models based on copulas by \cite{cuietal2019} 
and \cite{foketal2020}. Although these models can also be flexible regarding the contemporaneous correlation, this is difficult to assess since explicit forms for the correlation structure are not provided. This is further discussed in Subsection \ref{other_ingarch}.

The remainder of the paper is organized in the following way. In Section \ref{def_prop}, we define our proposed bivariate conditional Poisson INGARCH model, establish the properties of stationarity and ergodicity of the process, and compare it with existing bivariate INGARCH models. Section \ref{mle} is devoted to the estimation of the parameters via the conditional maximum likelihood method. Furthermore, we establish conditions to obtain consistency and asymptotic normality of the proposed estimators. Simulation studies are conducted to assess the finite-sample performance of the proposed estimators in Section \ref{simulation}. We also compare methods to obtain standard errors of the parameters, including asymptotic based methods and a bootstrap alternative. Hypothesis testing and simulated results involving the cross-correlation parameter are also addressed. A full data analysis of the bivariate counts of hepatitis which extracts the contemporaneous correlation in two nearby Brazilian cities is presented in Section \ref{application}. This empirical illustration demonstrates the importance of considering a bivariate model allowing for a wide range of contemporaneous correlation in real-life applications. Concluding remarks and future research are discussed in Section \ref{conclusion}.

\section{Bivariate conditional Poisson INGARCH process}\label{def_prop}

We begin by presenting the bivariate conditional Poisson (BCP) distribution introduced by \cite{berplu2004}. We say that a random vector $(Z_1,Z_2)$ follows a bivariate conditional Poisson distribution with parameters $\lambda_1,\lambda_2>0$ and $\phi\in\mathbb R$ if it satisfies the stochastic representation: $Z_1\sim\mbox{Poisson}(\lambda_1)$ and $Z_2|Z_1=z_1\sim\mbox{Poisson}(\mu_2e^{\phi z_1})$, where $\mu_2\equiv \lambda_2\exp\{-\lambda_1(e^\phi-1)\}$. We denote $(Z_1,Z_2)\sim\mbox{BCP}(\lambda_1,\lambda_2,\phi)$. Note that the marginal of $Z_2$ is not Poisson but mixed Poisson distributed. The mean and variance of $Z_2$ are given by $E(Z_2)=\lambda_2$ and  $\mbox{Var}(Z_2)=\lambda_2+\lambda_2^2\left\{\exp(\lambda_1(e^\phi-1)^2)-1\right\}$, respectively. As expected from the definition, $Z_2$ is overdispersed (variance greater than mean). Evidently, the marginal moments of $Z_1$ are obtained from the Poisson ones.

\begin{remark}
	In this paper, we develop a different parameterization of the BCP distribution with $\lambda_2$ being the marginal mean of $Z_2$, differently from \cite{berplu2004}. This will be important for the definition of our bivariate INGARCH process in terms of mean parameters and has the advantage of being easier to establish first-order stationarity of the bivariate count process.
\end{remark}

The joint probability function of $(Z_1,Z_2)$, say $p(x,y)\equiv P(Z_1=x,Z_2=y)$, is given by
\begin{eqnarray}\label{jointpf}
p(x,y)=\dfrac{\lambda_1^x\lambda_2^y}{x!y!}\exp\bigg\{-\lambda_1\left(1+y(e^\phi-1)\right)-\lambda_2\exp\left\{-\lambda_1(e^\phi-1)+\phi x\right\}+\phi xy\bigg\},
\end{eqnarray}
for $x,y\in\mathbb\{0,1,2,\ldots\}$. Joint moments for the BCP distribution are given in \cite{berplu2004}. The covariance between $Z_1$ and $Z_2$ is $\mbox{cov}(Z_1,Z_2)=\lambda_1\lambda_2(e^\phi-1)$ and therefore the correlation takes the form
\begin{eqnarray}\label{correlation}
\mbox{corr}(Z_1,Z_2)=(e^\phi-1)\sqrt{\dfrac{\lambda_1\lambda_2}{1+\lambda_2(e^{\lambda_1(e^\phi-1)^2}-1)}}.
\end{eqnarray}

\begin{remark}
	The parameter $\phi$ controls the dependence of the model. For $\phi=0$, $\phi>0$ and $\phi<0$, we have respectively independence, positive and negative correlations. Another remarkable point is that this parameter does not have restrictions depending on the means, in contrast with the previous bivariate models considered for constructing bivariate INGARCH models. This will enable us to deal with highly correlated bivariate count time series. 
\end{remark}

Figure \ref{bcp_corr} illustrates the features of the model. For fixed values of $\lambda_1$ and $\lambda_2$, the cross-correlation value is high when $\phi$ is close to 0, where positive small values of $\phi$ imply a high positive correlation while small negative values imply a high negative correlation. From this figure, we can see that the BCP distribution accommodates a wide range of cross-correlation for various values of the mean parameters.

\begin{remark}
	We now obtain explicitly the maximum and minimum points of (\ref{correlation}) as function of $\phi$, which can be expressed in terms of the Lambert function $W_k(x)$, with $k\in\mathbb Z$ and $x\in\mathbb R$ \citep{coretal1996}. This point has not been discussed in \cite{berplu2004}. By taking the first derivative of (\ref{correlation}) with respect to $\phi$ and equating 0, we obtain that $z e^z = e^{-1}(\lambda_2^{-1}-1)$, where $z = \lambda_1(e^\phi -1)^2 -1$. From the results given in \cite{coretal1996}, we have that $z=W_0(e^{-1}(\lambda_2^{-1}-1))$ when $\lambda_2\leq 1$; $W_0(\cdot)$ is known as the principal branch of the Lambert function. For $\lambda_2>1$, we obtain the real solutions $z=W_0(e^{-1}(\lambda_2^{-1}-1))$ and $z=W_{-1}(e^{-1}(\lambda_2^{-1}-1))$. Hence, explicit solutions in terms of $\phi$ can be obtained as well as the theoretical range of correlation.
\end{remark}

\begin{figure}
	\centering
	\includegraphics[width=0.6\textwidth]{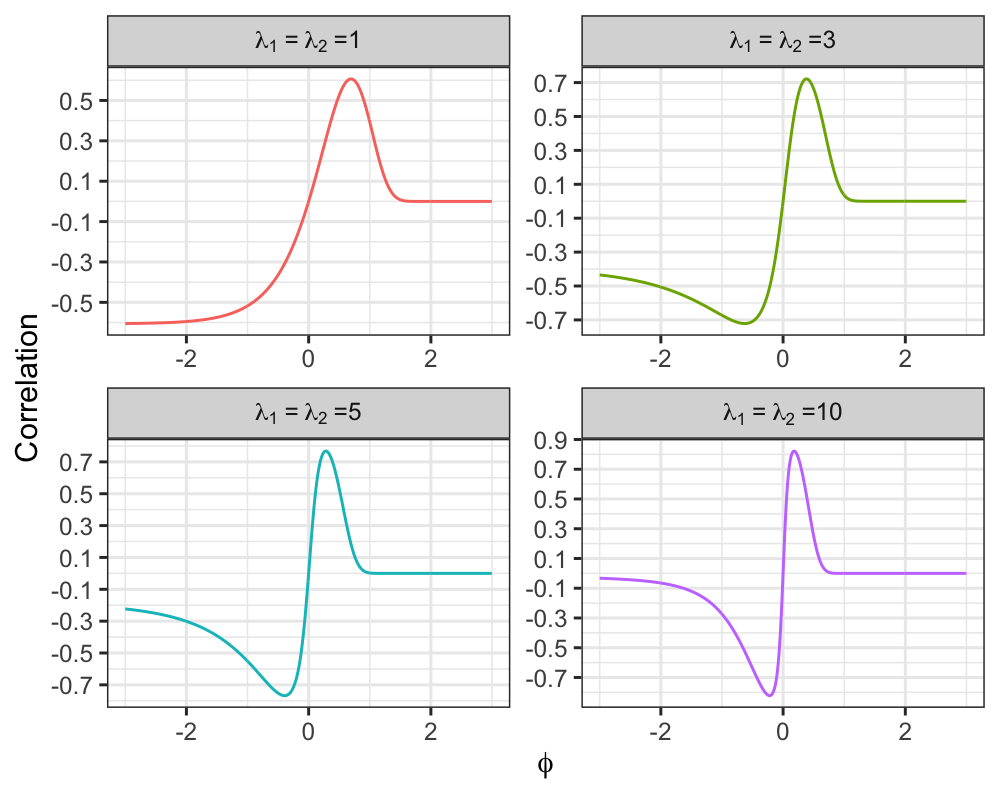}
	\caption{Cross-correlation for the bivariate conditional Poisson distribution as function of $\phi$ for some values of $\lambda_1$ and $\lambda_2$.}\label{bcp_corr}
\end{figure}

With the above bivariate conditional Poisson distribution, we can define our proposed bivariate count process with a flexible range of contemporaneous correlation as follows.

\begin{definition}\label{bcpingarch_def}(BCP-INGARCH process).
	Let ${\bf Y}_t=(Y_{t1},Y_{t2})^\top$ be a bivariate count time series where $t\geq 1$. We say that $\{\bf Y_t\}_{t\geq1}$ is a bivariate conditional Poisson INGARCH(1,1) process if it satisfies
	\begin{eqnarray}\label{BCP_equation}
	{\bf Y}_t|\mathcal F_{t-1}&\sim&\mbox{BCP}(\lambda_{1t},\lambda_{2t},\phi),\nonumber\\
	\boldsymbol{\lambda}_{t}\equiv E({\bf Y}_t|\mathcal F_{t-1})&=&\boldsymbol\omega+{\bf A}\boldsymbol\lambda_{t-1}+{\bf B}{\bf Y}_{t-1}, 
	\end{eqnarray}
	where $\mathcal F_{t-1}=\sigma(\bm Y_{t-1},\ldots,\bm Y_1,\bm\lambda_1)$, for $t\geq2$, $\boldsymbol\omega=(\omega_1,\omega_2)^\top\in\mathbb R_+^2$ is the intercept vector, ${\bf A}=\{\alpha_{ij}\}_{i,j=1,2}$ and ${\bf B}=\{\beta_{ij}\}_{i,j=1,2}$ are two $2\times2$ matrices with non-negative entries/parameters, $\phi\in\mathbb R$ is the contemporaneous dependence parameter, and $\bm Y_1|\mathcal F_0\sim \mbox{BCP}(\lambda_{11},\lambda_{21},\phi)$, with $\mathcal F_0=\sigma(\boldsymbol\lambda_1)$ and $\boldsymbol\lambda_1=(\lambda_{11},\lambda_{21})^\top$.
\end{definition}

\begin{remark}
	Note that we have defined above explicitly the conditional distribution of $\bm Y_1|\mathcal F_0$ following our baseline bivariate conditional Poisson distribution. In the existing bivariate INGARCH models, such an assumption is not mentioned but it is implicitly used for showing, for example, that the bivariate Markov process $\{\boldsymbol\lambda_t\}_{t\geq1}$ is an e-chain, which is, in particular, important to obtain desirable theoretical properties of the bivariate count process $\{{\bf Y}_t\}_{t\geq1}$.
\end{remark}

We have that $E({\bf Y}_t)=E(\boldsymbol{\lambda}_{t})=\boldsymbol\omega+{\bf A}E(\boldsymbol\lambda_{t-1})+{\bf B}E({\bf Y}_{t-1})$, and under first-order stationarity, we obtain $E({\bf Y}_t)=({\bf I}-{\bf A}-{\bf B})^{-1}\boldsymbol\omega$, for $t\geq1$, 	
where ${\bf I}$ is the identity matrix. Detailed discussion on the stationarity and ergodicity of the count process is provided in the following subsection.

\subsection{Stability theory}\label{ST}

We now introduce some matrix notations to state the main results on stability theory for the proposed bivariate count process. We follow the notation used in \cite{liu2012}. For a matrix $\bm J\in\mathbb C^{m\times n}$ and $p\in[1,\infty]$, we denote $\|\bm J\|_p=\max_{\bm x\neq0}\{\|\bm J \bm x\|_p/\|\bm x\|_p: \bm x\in\mathbb C^n\}$, with $\|\bm x\|_p=\left(\sum_{i=1}^n|x_i|^p\right)^{1/p}$ for $p\in[1,\infty)$ and $\|\bm x\|_\infty=\max_{1\leq i\leq n}|x_i|$ for $p=\infty$ being the $p$-norm of the vector $\bm x$. For $p=1$ and $p=\infty$, we have respectively that $\|\bm J \|_p=\max_{1\leq j\leq n}\sum_{i=1}^m|\bm J_{ij}|$ and $\|\bm J \|_\infty=\max_{1\leq i\leq m}\sum_{j=1}^n|\bm J_{ij}|$, with $\bm J_{ij}$ denoting the $(i,j)$-th element of $\bm J$, $i=1,\ldots,n$ and $j=1,\ldots,m$. Here, $\rho(\bm A)$ denotes the largest absolute eigenvalue of the matrix $\bm A$. Further, define $\mathcal H$ to be the class of continuous real functions with compact support on $[0,\infty)\times[0,\infty)$. We now have all the ingredients to establish an important result about the joint random bivariate vector $\{\bm \lambda_t\}_{t\geq1}$ as follows.

\begin{theorem}\label{e-chain}
	If $\|\bm A\|_p<1$ for some $p\in[1,\infty]$, then $\{\bm \lambda_t\}_{t\geq1}$ is an e-chain.
\end{theorem}

\begin{proof}
	By $\{\bm \lambda_t\}_{t\geq1}$ be an e-chain it means that for any function $g\in\mathcal H$ and for every $\epsilon>0$, there exists $\eta>0$ such that $\|\bm x-\bm z\|<\eta$ implies $\big|E\left(g(\bm\lambda_k)|\bm\lambda_0=\bm x\right)-E\left(g(\bm\lambda_k)|\bm\lambda_0=\bm z\right)\big|<\epsilon$ for all $k\geq1$, where $\|\cdot\|$ is some norm and $\bm z,\bm x\in(0,\infty)\times (0,\infty)$. 
	
	Let $\epsilon>0$, $\bm z_1,\bm x_1\in(0,\infty)\times (0,\infty)$, and $g\in\mathcal H$. Let us assume that $|g|<1$  without loss of generality (since its support is compact). We begin by dealing with the case $k=1$. 
	
	Denote $p_{\bm \lambda}(m,n)$ be the joint probability function of a bivariate conditional Poisson given by (\ref{jointpf}) with mean vector $\bm \lambda$ and correlation parameter $\phi$. By using the triangle inequality, we have that
	\begin{eqnarray}\label{double_sum}
	&&	\big|E\left(g(\bm\lambda_1)|\bm\lambda_0=\bm x_1\right)-E\left(g(\bm\lambda_1)|\bm\lambda_0=\bm z_1\right)\big|=\nonumber\\
	&&	\sum_{m=0}^\infty\sum_{n=0}^\infty\big|g(\bm\omega+\bm A\bm x_1+\bm B (m,n)^\top)p_{\bm x_1}(m,n)-g(\bm\omega+\bm A\bm z_1+\bm B (m,n)^\top)p_{\bm z_1}(m,n)\big|\leq\nonumber\\
   	&&	\sum_{m=0}^\infty\sum_{n=0}^\infty p_{\bm x_1}(m,n)\big|g(\bm\omega+\bm A\bm x_1+\bm B (m,n)^\top)-g(\bm\omega+\bm A\bm z_1+\bm B (m,n)^\top)\big|+\nonumber\\
	&&	\sum_{m=0}^\infty\sum_{n=0}^\infty\big|g(\bm\omega+\bm A\bm z_1+\bm B (m,n)^\top)\big| \big|p_{\bm x_1}(m,n)-p_{\bm z_1}(m,n)\big|.
	\end{eqnarray}	

Let us now to find a superior bound for the second double summation in (\ref{double_sum}). Let $p_{\bm x_1}(n|m)$ and $p_{\bm x_{11}}(m)$ be the probability functions of a Poisson distribution with respective means $x_{12}\exp\{\phi m-x_{11}(e^\phi-1)\}$ and $x_{11}$ (defined similarly for $p_{\bm z_1}(n|m)$ and $p_{\bm z_{11}}(m)$). Then, it follows that 
\begin{eqnarray}\label{aux1}
&&\big|p_{\bm x_1}(m,n)-p_{\bm z_1}(m,n)\big|=\big|p_{\bm x_1}(n|m)p_{\bm x_{11}}(m)-p_{\bm z_1}(n|m)p_{\bm z_{11}}(m)\big|\leq\\
&& \max\{p_{\bm x_1}(n|m),p_{\bm z_1}(n|m)\}\big|p_{\bm x_{11}}(m)-p_{\bm z_{11}}(m)\big|\leq(p_{\bm x_1}(n|m)+p_{\bm z_1}(n|m)) \big|p_{\bm x_{11}}(m)-p_{\bm z_{11}}(m)\big|,\nonumber
\end{eqnarray}
where we have used in the first inequality the fact that if $a,b,c,d\geq0$ then $|ab-cd|\leq|a-c|\max(b,d)$.
Using these results and the fact that $|g|<1$, we have that the second double summation in (\ref{double_sum}) is bounded above by
\begin{eqnarray}\label{aux2}
	\sum_{m=0}^\infty\sum_{n=0}^\infty\big|p_{\bm x_1}(m,n)-p_{\bm z_1}(m,n)\big|&\leq& 2 \sum_{m=0}^\infty \big|p_{\bm x_{11}}(m)-p_{\bm z_{11}}(m)\big|\\
   &\leq& 4(1-e^{-|x_{11}-z_{11}|})\leq 4(1-e^{-2\|\bm x_1-\bm z_1\|_p}),\nonumber
\end{eqnarray}	
where the second inequality is due to \cite{liu2012}, and the third one follows by using that $|x_{11}-z_{11}|\leq\|\bm x_1-\bm z_1\|_1\leq 2^{1-1/p}\|\bm x_1-\bm z_1\|_p\leq 2\|\bm x_1-\bm z_1\|_p$, for all $p\in[1,\infty]$.

The superior bound for the first double summation in (\ref{double_sum}) follows as discussed by \cite{liu2012}. Since $g$ is a continuous function, we can take $\epsilon'>0$ and $\eta>0$ sufficiently small such as $\epsilon'+\dfrac{8\eta}{1-\|A\|_p}<\epsilon$ and $|g(\bm x_1)-g(\bm z_1)|<\epsilon'$ whenever $\|\bm x_1-\bm z_1\|_p<\eta$
for some $p\in[1,\infty]$. Hence, it follows that $\|\bm\omega+\bm A\bm x_1+\bm B(m,n)^\top-(\bm\omega+\bm A\bm z_1+\bm B(m,n)^\top)\|_p=\|\bm A(\bm x_1-\bm z_1)\|_p\leq \|\bm A\|_p\|\bm x_1-\bm z_1\|_p\leq \|\bm x_1-\bm z_1\|_p\leq \eta$, where the second inequality follows by the assumption that $\|\bm A\|_p<1$. This implies that $\big|g(\bm\omega+\bm A\bm x_1+\bm B (m,n)^\top)-g(\bm\omega+\bm A\bm z_1+\bm B (m,n)^\top)\big|<\epsilon'$ and therefore
\begin{eqnarray*}
		\sum_{m=0}^\infty\sum_{n=0}^\infty p_{\bm x_1}(m,n)\big|g(\bm\omega+\bm A\bm x_1+\bm B (m,n)^\top)-g(\bm\omega+\bm A\bm z_1+\bm B (m,n)^\top)\big|\leq \epsilon', 
\end{eqnarray*}	
since $\sum_{m=0}^\infty\sum_{n=0}^\infty p_{\bm x_1}(m,n)=1$. By combining the above results, we get
\begin{eqnarray}\label{case_k=1}
\big|E\left(g(\bm\lambda_1)|\bm\lambda_0=\bm x_1\right)-E\left(g(\bm\lambda_1)|\bm\lambda_0=\bm z_1\right)\big|\leq \epsilon'+4(1-e^{-2\|\bm x_1-\bm z_1\|_p}).
\end{eqnarray}

Using the case $k=1$ provided in (\ref{case_k=1}) and induction for general $k$, the same steps of \cite{liu2012} (Chapter 4, page 109) yield that
\begin{eqnarray*}
&&\big|E\left(g(\bm\lambda_k)|\bm\lambda_0=\bm x_1\right)-E\left(g(\bm\lambda_k)|\bm\lambda_0=\bm z_1\right)\big|\leq \epsilon'+4\sum_{s=0}^{k-1}(1-e^{-2\|\bm A\|_p^s\|\bm x_1-\bm z_1\|_p})\leq\\
&& \epsilon'+4\sum_{s=0}^{\infty}(1-e^{-2\|\bm A\|_p^s\|\bm x_1-\bm z_1\|_p})\leq \epsilon'+8\|\bm x_1-\bm z_1\|_p\sum_{s=0}^{\infty}\|\bm A\|_p^s\leq \epsilon'+\dfrac{8\eta}{1-\|\bm A\|_p}\leq \epsilon,
\end{eqnarray*}
where we have used that $1-e^{-x}\leq x$ for all $x\geq0$ in the third inequality and that $\|\bm A\|_p<1$ in the forth inequality. This completes the proof of the desired result.
\end{proof}

\begin{remark}
	A key ingredient to establish the e-chain property of $\{\bm \lambda_t\}_{t\geq1}$ is inequality (\ref{aux2}). Our approach given in (\ref{aux1}) only uses that one of the marginals (of the baseline bivariate distribution) is Poisson distributed. We mean the result holds whatever is the conditional distribution $p_\cdot(m|n)$. Therefore, the argument used here is simpler and more general than those used in \cite{liu2012}, \cite{cuizhu2018}, and \cite{cuietal2019}, where other bivariate Poisson distributions are considered.
\end{remark}

With Theorem \ref{e-chain} at hand, we can use the results given in \cite{liu2012} to obtain conditions ensuring stationarity and ergodicity for $\{(\bm Y_t,\bm \lambda_t)\}_{t\geq1}$. Under the conditions $\rho(\bm A+\bm B)<1$ and $\|\bm A\|_p<1$ for some $p\in[1,\infty]$, $\{(\bm Y_t,\bm \lambda_t)\}_{t\geq1}$ has a unique stationary solution. If $\|\bm A\|_p+2^{1-1/p}\|\bm B\|_p<1$ for some $p\in[1,\infty]$,  $\{(\bm Y_t,\bm \lambda_t)\}_{t\geq1}$ has a unique stationary and ergodic solution.

\begin{remark}
	It is noteworthy that \cite{liu2012} has established seminal results on the stationarity and ergodicity of bivariate INGARCH models. These results have been used for instance by \cite{leeetal2018}, \cite{cuizhu2018}, and \cite{cuietal2019}.
\end{remark}

\subsection{Existing bivariate INGARCH models and their limitations}\label{other_ingarch}

We here discuss the existing bivariate INGARCH models and present some problems regarding the contemporaneous correlation, which relies on the ability for capturing the dependence of the baseline bivariate discrete distribution. The model in \cite{liu2012} and \cite{leeetal2018} is defined by $\bm Y_t|\mathcal F_{t-1}\sim\mbox{BP}^\star(\lambda_{1t},\lambda_{2t},\phi)$ and the dynamics for $\bm\lambda_t$ as in (\ref{BCP_equation}), where $\mbox{BP}^\star$ stands for the bivariate Poisson distribution obtained via the trivariate reduction method, assuming the form
\begin{eqnarray}\label{bivpoisI}
P(Y_{1t}=y_{1},Y_{2t}=y_{2}|\mathcal F_{t-1})&=&e^{-(\lambda_{1t}+\lambda_{2t}-\phi)}\dfrac{(\lambda_{1t}-\phi)^{y_1}(\lambda_{2t}-\phi)^{y_2}}{y_1!y_2!}\times\\&&\sum_{s=0}^{\min(y_1,y_2)}\binom{y_1}{s}\binom{y_2}{s}s!\left(\dfrac{\phi}{(\lambda_{1t}-\phi)(\lambda_{2t}-\phi)}\right)^s,\quad y_1,y_2\in\{0,1,\ldots\},\nonumber
\end{eqnarray}	
with $\phi=\mbox{cov}(Y_{1t},Y_{2t})\in[0,\min(\lambda_{1t},\lambda_{2t}))$ deterministic and does not depending on $t$. To ensure this last condition, the authors assumed that $\phi<\min(a_1,a_2)$, with $(a_1,a_2)^\top=(\bm I-\bm A)^{-1}\bm\omega$ since $\bm\lambda_t\geq(\bm I-\bm A)^{-1}\bm\omega$ $\forall t\geq1$ when $\rho(\bm A)<1$ \citep{liu2012}.
The bivariate distribution in (\ref{bivpoisI}) has Poisson marginals conditional on $\mathcal F_{t-1}$. One of the limitations of this model is that it does not allow for negative contemporaneous correlation and the parameter value that corresponds to independence lies on the boundary of the parameter space. As argued by \cite{berplu2004}, this bivariate Poisson distribution also cannot accommodate higher values of positive correlation, especially for large values of the marginal means (see Eq. (12) from that paper). All these restrictions on the parameter $\phi$ imply compromised practical applicability, where a broad range of correlation is required.

We now discuss the model in \cite{cuizhu2018} given by $\bm Y_t|\mathcal F_{t-1}\sim\mbox{BP}^\dag(\lambda_{1t},\lambda_{2t},\phi)$, with $\bm\lambda_t$ satisfying the dynamics as in (\ref{BCP_equation}). Here $\mbox{BP}^\dag$ denotes the bivariate Poisson distribution with probability function
\begin{eqnarray}\label{bivpoisII}
P(Y_{1t}=y_{1},Y_{2t}=y_{2}|\mathcal F_{t-1})=e^{-(\lambda_{1t}+\lambda_{2t})}\dfrac{\lambda_{1t}^{y_1}\lambda_{2t}^{y_2}}{y_1!y_2!}\left\{1+\phi(e^{-y_1}-e^{-c\lambda_{1t}})(e^{-y_2}-e^{-c\lambda_{2t}})\right\},
\end{eqnarray}	
for $y_1,y_2\in\{0,1,\ldots\}$, where $c=1-e^{-1}$. The parameter space related to $\phi$ is incorrectly stated in that paper. The bivariate Poisson distribution in (\ref{bivpoisII}) belongs to a more general class of distributions proposed by \cite{sar1966} and the correct range of $\phi$ is well-known in the literature. For instance, the correct range can be found in Subsection 2.2 (after Eq. (1)) from \cite{hoflei2012}; see also \cite{lee1996}. Using that result in the particular case given in (\ref{bivpoisII}), we obtain that 
\begin{eqnarray}\label{bivpoisIIrange}
\dfrac{-1}{\max\{e^{-c(\lambda_{1t}+\lambda_{2t})},(1-e^{-c\lambda_{1t}})(1-e^{-c\lambda_{2t}})\}}<\phi<
\dfrac{1}{\max\{e^{-c\lambda_{1t}}(1-e^{-c\lambda_{2t}}),e^{-c\lambda_{2t}}(1-e^{-c\lambda_{1t}})\}},
\end{eqnarray}	
in contrast with the wrong range considered by \cite{cuizhu2018} $|\phi|<\dfrac{1}{(1-e^{-c\lambda_{1t}})(1-e^{-c\lambda_{2t}})}$. The obvious implication of this incorrect bound is that the model there is not well-defined. Moreover, note that (\ref{bivpoisIIrange}) needs to be deterministic and independent of $t$ as done by \cite{leeetal2018} and therefore an additional restriction is necessary. Such restriction is not clearly discussed by \cite{cuizhu2018}. It is worth noting that, although both negative and positive contemporaneous correlation are allowed in that model, its range is extremely limited due to the simplex structure in (\ref{bivpoisII}), as discussed by \cite{cuietal2019}. 

In the papers above, one can observe two major problems regarding the range of cross-correlation: (1) natural restriction of the parameter space due to the baseline bivariate discrete distribution; (2) bounds of the parameter space depending on the conditional means, which are driven by stochastic dynamics. Our proposed model deals directly with both issues and naturally overcomes the limitations because the parameter space of $\phi$ is $\mathbb R$-valued (so it does not depend on the means) and the baseline distribution allows for a broad range of correlation. 

Further related advancements on multivariate INGARCH models are based on copulas. \cite{cuietal2019} proposed a bivariate Poisson INGARCH model having a copula structure as an alternative to the simplex form in (\ref{bivpoisII}). Although this provides greater flexibility for modeling dependence, it adds a cumbersome normalizing constant in the joint probability function which involves a double infinite summation. Further, it is difficult to assess the possible range of correlation. Another related work is due to \cite{foketal2020}, where a multivariate INGARCH model is elegantly introduced based on a latent copula approach. The paper studied the stochastic properties of the multivariate count process and estimate parameters via a quasi-likelihood approach, which is equivalent to modeling the multivariate counts under the assumption of contemporaneous independence (see Section 4 from that paper). A trick point that arises in this approach is the estimation of the parameter responsible for controlling the dependence. Further, like the model proposed by \cite{cuietal2019}, it is to hard assess the possible range of contemporaneous correlation.

As demonstrated in the following section, our likelihood function assumes a very simple form, and the parameters are jointly estimated (including the cross-correlation parameter) via the conditional maximum likelihood method.

\section{Statistical inference and asymptotic results}\label{mle}

To estimate the parameters in the BCP-INGARCH(1,1) model, we consider the conditional maximum likelihood (CML) approach. Here, we also determine conditions that ensure asymptotic normality of the CML estimators. Following \cite{hei2003} and \cite{leeetal2018}, we consider $\bf A$ (see Equation (\ref{BCP_equation})) to be a diagonal matrix with elements $\alpha_1$ and $\alpha_2$. As argued by the authors, besides reducing the number of parameters to be estimated, this makes the study of the asymptotic properties of the maximum likelihood estimators feasible. Denote by $\boldsymbol\theta=(\mbox{vec}({\bf A}),\mbox{vec}({\bf B}), \boldsymbol\omega^\top, \phi)^\top$ the parameter vector, where $\mbox{vec}({\bf A})\equiv (\alpha_{1},\alpha_{2})^\top$ and $\mbox{vec}({\bf B})\equiv (\beta_{11},\beta_{12},\beta_{21},\beta_{22})^\top$. 

Due to the Markovian property of the process, the conditional joint probability function of $\{{\bf Y}_t\}_{t=2}^n$ is given by ${\bf Y}_1={\bf y}_1$ is given by $P({\bf Y}_n={\bf y}_n,{\bf Y}_{n-1}={\bf y}_{n-1},\ldots,{\bf Y}_2={\bf y}_2|{\bf Y}_1={\bf y}_1)=\prod_{t=2}^n P({\bf Y}_t={\bf y}_t|{\bf Y}_{t-1}={\bf y}_{t-1})$, where $n$ is the sample size. The conditional log-likelihood function, denoted by $\ell(\boldsymbol\theta)$, is
\begin{eqnarray*}
	\ell(\boldsymbol\theta)  \propto  \displaystyle\sum_{t=2}^n\bigg\{y_{1t}\log\lambda_{1t}+y_{2t}\log\lambda_{2t}-\lambda_{1t}\left(1+y_{2t}(e^\phi-1)\right)- \lambda_{2t}\exp\left\{-\lambda_{1t}(e^\phi-1)+\phi y_{1t}\right\}+\phi y_{1t}y_{2t}\bigg\},	
\end{eqnarray*}	
with ${\bf y}_t=(y_{1t},y_{2t})^\top$ denoting the observed value of ${\bf Y}_t=(Y_{1t},Y_{2t})^\top$, for $t=1,\ldots,n$. The conditional maximum likelihood estimator (CMLE) of $\boldsymbol\theta$ is given by  $\widehat{\boldsymbol\theta}=\mbox{argmax}_{{\boldsymbol\theta}\in{\boldsymbol\Theta}}\ell({\boldsymbol\theta})$, with ${\boldsymbol\Theta}$ denoting the parameter space. To derive  $\widehat{\boldsymbol{\theta}}$, we employ the numeric optimization routine provided by the Stan software \citep{stan} through the \verb|R| \citep{R} package \verb|rstan|. Stan is a platform for high performance statistical computation where fast results are achieved through compilation in C++. Additionally to the efficiency gain, we have found in this and other works that Stan's numeric maximization can yield superior results in comparison to the standard \verb|optim| command in \verb|R|.
 
The associated score function to the log-likelihood $\ell(\boldsymbol\theta)$ is denoted by $U(\boldsymbol\theta)=\partial\ell(\boldsymbol\theta)/\partial\boldsymbol\theta$, where its components are given by
\begin{eqnarray*}\label{score}
	\dfrac{\partial\ell(\boldsymbol\theta)}{\partial\omega_j}=\sum_{t=2}^n S_{jt}(\boldsymbol\theta)\dfrac{\partial\lambda_{jt}}{\partial \omega_{j}},\quad
	\dfrac{\partial\ell(\boldsymbol\theta)}{\partial\alpha_{j}}=\sum_{t=2}^n S_{jt}(\boldsymbol\theta)\dfrac{\partial\lambda_{jt}}{\partial \alpha_{j}},\quad
	\dfrac{\partial\ell(\boldsymbol\theta)}{\partial\beta_{kj}}=\sum_{t=2}^n S_{jt}(\boldsymbol\theta)\dfrac{\partial\lambda_{jt}}{\partial \beta_{kj}},\quad
	\dfrac{\partial\ell(\boldsymbol\theta)}{\partial\phi}=\sum_{t=2}^n S_{3t}(\boldsymbol\theta),
\end{eqnarray*}
with
\begin{eqnarray*}\label{Ss}
	S_{1t}(\boldsymbol\theta)&=&\dfrac{y_{1t}}{\lambda_{1t}}-1+(e^\phi-1)\big(\lambda_{2t}\exp\left\{-\lambda_{1t}(e^\phi-1)+\phi y_{1t}\right\}-y_{2t}\big),\\
	S_{2t}(\boldsymbol\theta)&=&\dfrac{y_{2t}}{\lambda_{2t}}-\exp\left\{-\lambda_{1t}(e^\phi-1)+\phi y_{1t}\right\},\\
	S_{3t}(\boldsymbol\theta)&=&-y_{2t}\lambda_{1t}e^\phi-\lambda_{2t}\exp\left\{-\lambda_{1t}(e^\phi-1)+\phi y_{1t}\right\}(y_{1t}-\lambda_{1t}e^\phi)+y_{1t}y_{2t},
\end{eqnarray*}
and
\begin{eqnarray*}
	\dfrac{\partial\lambda_{jt}}{\partial \omega_{j}}=1+\alpha_j\dfrac{\partial\lambda_{j\,t-1}}{\partial\omega_j},\quad 
	\dfrac{\partial\lambda_{jt}}{\partial \alpha_{j}}=\lambda_{j\,t-1}+\alpha_j\dfrac{\partial\lambda_{j\,t-1}}{\partial\alpha_j},\quad 
	\dfrac{\partial\lambda_{jt}}{\partial \beta_{jk}}=y_{k\,t-1}+\alpha_j\dfrac{\partial\lambda_{j\,t-1}}{\partial\beta_{jk}},
\end{eqnarray*}	
for $k,j=1,2$ and $t=2,\ldots,n$.

The next result gives us some properties of the score function useful to establish the asymptotic normality of the CML estimators.

\begin{theorem}
	We have that $\{{\bm U}(\boldsymbol\theta);\mathcal F_{t-1}\}$ is a martingale difference sequence. Further, ${\bm U}(\boldsymbol\theta)$ satisfies the information matrix equality
	\begin{eqnarray}\label{bartlett_identity}
		-E(\nabla {\bm U}(\boldsymbol\theta))=E({\bm U}(\boldsymbol\theta) {\bm U}(\boldsymbol\theta)^\top).
	\end{eqnarray}	
\end{theorem}

\begin{proof}
		
	We have that $E(Y_{jt}|\mathbb F_{t-1})=\lambda_{jt}$, for $j=1,2$, $E(\exp\{\phi Y_{1t}\}|\mathcal F_{t-1})=\exp\{\lambda_{1t}(e^\phi-1)\}$, $E(Y_{1t}\exp\{\phi Y_{1t}\}|\mathcal F_{t-1})=\lambda_{1t}\exp\{\phi+\lambda_{1t}(e^\phi-1)\}$, and $E(Y_{1t}Y_{2t}|\mathcal F_{t-1})=\lambda_{1t}\lambda_{2t}e^\phi$.

	Using the above expectations, we obtain that $E(S_{jt}(\bm\theta)|\mathcal F_{t-1})=0$ for $j=1,2$ and $t=2,\ldots,n$. Moreover, $\dfrac{\partial\lambda_{jt}}{\partial\bm\theta}$ is $\mathcal F_{t-1}$-mensurable. These results implies that $E(\bm U(\bm\theta)|\mathcal F_{t-1})=0$ almost surely. Therefore,  $\{\bm U(\bm\theta);\mathcal F_{t-1}\}$ is a martingale difference sequence.

	Let us now to show that the identity (\ref{bartlett_identity}) holds. Consider the case where the derivatives are taken with respect to $\alpha_1$. The remaining cases follow in a similar fashion and therefore are omitted. We have that
	\begin{eqnarray*}
		\dfrac{\partial^2\ell(\bm\theta)}{\partial\alpha_1^2}=\sum_{t=2}^n\dfrac{\partial S_{1t}(\bm\theta)}{\partial\alpha_1}\left(\dfrac{\partial\lambda_{1t}}{\partial\alpha_1}\right)^2+\sum_{t=2}^n S_{1t}(\bm\theta)\dfrac{\partial^2\lambda_{1t}}{\partial\alpha_1^2}.
	\end{eqnarray*}	
	By using that $E(\partial S_{1t}(\bm\theta)/\partial\alpha_1|\mathcal F_{t-1})=\dfrac{1}{\lambda_{1t}}+(e^\phi-1)^2\lambda_{2t}$ and $E(S_{1t}(\bm\theta)|\mathcal F_{t-1})=0$, we obtain that
	\begin{eqnarray}\label{eq_I}
		E\left(-\frac{\partial^2\ell(\bm\theta)}{\partial\alpha_1^2}\right)=\sum_{t=2}^nE\left[\left(\dfrac{1}{\lambda_{1t}}+(e^\phi-1)^2\lambda_{2t}\right)\left(\dfrac{\partial\lambda_{1t}}{\partial\alpha_1}\right)^2\right].
	\end{eqnarray}	
	
	On the other hand, it follows that
	\begin{eqnarray}\label{eq_II}
		E\left(\left(\frac{\partial\ell(\bm\theta)}{\partial\alpha_1}\right)^2\right)=\sum_{t=2}^nE\left[S^2_{1t}(\bm\theta)\left(\dfrac{\partial\lambda_{1t}}{\partial\alpha_1}\right)^2\right]+\sum_{t\neq t'}^nE\left[S_{1t}(\bm\theta)S_{1t'}(\bm\theta)\dfrac{\partial\lambda_{1t}}{\partial\alpha_1}\dfrac{\partial\lambda_{1t'}}{\partial\alpha_1}\right].
	\end{eqnarray}
	
	 Assume that $t>t'$ without loss of generality. Then, $$E\left[S_{1t}(\bm\theta)S_{1t'}(\bm\theta)\dfrac{\partial\lambda_{1t}}{\partial\alpha_1}\dfrac{\partial\lambda_{1t'}}{\partial\alpha_1}\right]=E\left[S_{1t'}(\bm\theta)\dfrac{\partial\lambda_{1t}}{\partial\alpha_1}\dfrac{\partial\lambda_{1t'}}{\partial\alpha_1}E\left(S_{1t}(\bm\theta)|\mathcal F_{t-1}\right)\right]=0$$
	 since $S_{1t'}(\bm\theta)\partial\lambda_{1t}/\partial\alpha_1\partial\lambda_{1t'}/\partial\alpha_1$ is $\mathcal F_{t-1}$-mensurable and $E\left(S_{1t}(\bm\theta)|\mathcal F_{t-1}\right)=0$. 
	 
	 Define $\mathcal G_t\equiv \sigma(Y_{1,t},\bm Y_{t-1},\ldots,\bm Y_1,\bm \lambda_1)$. Thus, $Y_{2t}|\mathcal G_t\sim\mbox{Poisson}(\lambda_{2t}\exp\{-\lambda_{1t}(e^\phi-1)+\phi y_{1t}\})$. Hence, it follows that
	 \begin{eqnarray*}
	 	E(S^2_{1t}(\bm\theta)|\mathcal F_{t-1})&=&\mbox{Var}\left(\frac{Y_{1t}}{\lambda_{1t}}\Big|\mathcal F_{t-1}\right)+(e^\phi-1)^2E\left(\mbox{Var}\left(Y_{2t}\big|\mathcal G_t\right)\big|\mathcal F_{t-1}\right)\\
	 	&=&\dfrac{1}{\lambda_{1t}}+(e^\phi-1)^2E\left(\lambda_{2t}\exp\{-\lambda_{1t}(e^\phi-1)+\phi y_{1t}\}\big|\mathcal F_{t-1}\right)\\
	 	&=&\dfrac{1}{\lambda_{1t}}+(e^\phi-1)^2\lambda_{2t}.
	 \end{eqnarray*}	
 
\noindent Applying the above expectations in (\ref{eq_II}), we obtain that it equals (\ref{eq_I}), so proving the desired result.	
\end{proof}

The next point is to develop the asymptotic normality of the conditional maximum likelihood estimator, where some regularity conditions are necessary as follows.

\begin{assumption}\label{RCp}
	There exists $p\in[1,\infty]$ such as $\|\bm A\|_p+2^{1-1/p}\|\bm B\|_p<1$.
\end{assumption}

\begin{assumption}\label{compact}
	The true parameter value $\bm\theta_0$ is an interior point of $\bm\Theta$, with $\bm\Theta$ being a compact set.
\end{assumption}

\begin{remark}
	In the simulated and real data analyses, we consider that Assumption \ref{RCp} is in force with $p=1$.
\end{remark}

Expanding $\bm U(\widehat{\bm\theta})$ in Taylor's series around $\bm\theta_0$, we obtain that ${\bf0}=\bm U(\widehat{\bm\theta})=\bm U(\bm\theta_0)+(\widehat{\bm\theta}-\bm\theta_0)\nabla \bm U(\widetilde{\bm\theta})$, where $\widetilde{\bm\theta}$ belongs to the segment connecting the points $\widehat{\bm\theta}$ and ${\bm\theta}_0$. We rearrange the terms to obtain that
\begin{eqnarray}\label{expansion}
	\sqrt{n}(\widehat{\bm\theta}-\bm\theta_0)\left(-\dfrac{\nabla\bm U(\widetilde{\bm\theta})}{n}\right)=\dfrac{\bm U(\bm\theta_0)}{\sqrt{n}}.
\end{eqnarray}	

Under Assumption \ref{RCp}, the stationarity and ergodicity of $\{(\bm Y_t,\bm\lambda_t)\}_{t\geq1}$ implies in the same properties for $\bm U(\bm\theta_0)$. Hence, we use the Central Limit Theorem for Martingales \citep{halhey1980} and get that 
$\dfrac{\bm U(\bm\theta_0)}{\sqrt{n}}\stackrel{d}{\longrightarrow}N(\bm 0,\bm I(\bm\theta_0))$ as $n\rightarrow\infty$, where 
$\bm I(\bm\theta_0)$ is the Fisher information matrix which can be obtained as the limit in probability $\bm I(\bm\theta_0)=\mbox{plim}_{n\rightarrow\infty}\dfrac{1}{n}\displaystyle\sum_{t=2}^n E\left(-\nabla\bm U(\bm\theta_0) |\mathcal F_{t-1}\right)$.

Now, under Assumptions \ref{RCp} and \ref{compact}, we apply the Law of Large Number for Martingales and follow the steps of the proof of Proposition 5 by \cite{leeetal2016} (see also Lemma 6 by \cite{leeetal2018}), to obtain that $-\dfrac{\nabla\bm U(\widetilde{\bm\theta})}{n}\stackrel{a.s.}{\longrightarrow}\bm I(\bm\theta_0)$ as $n\rightarrow\infty$. By combining the above results in (\ref{expansion}), we have that $\sqrt{n}(\widehat{\bm\theta}-\bm\theta_0)\stackrel{d}{\longrightarrow} N(\bm0,\bm I^{-1}(\bm\theta_0))$.

With this asymptotic normality in hands, we can assess standard errors, construct confidence intervals for the parameters, and test the hypothesis of interest. In the next section, we provide some simulated results aiming at (i) study of the finite-sample performance of the CMLEs (ii) the evaluation of some strategies to get the standard errors of the estimates including a bootstrap approach; (iii) hypothesis testing for $H_0: \phi=0$ against $H_1: \phi\neq0$ (testing the presence of contemporaneous correlation). For this last aim, we compare the performance of the likelihood ratio and score tests. Due to the asymptotic normality of the conditional maximum likelihood estimators and the fact that $\phi=0$ does not belong to the boundary space, these statistics are asymptotically  $\chi_1^2$-distributed under the null hypothesis.

\section{Simulation studies}\label{simulation}

\subsection{Point estimation}

In this section, we assess the finite sample behavior of the conditional maximum likelihood estimators for the BCP-INGARCH(1,1) model through Monte Carlo simulations. In this study, 1000 replicas are used and the sample sizes $n = 200, 500$ are investigated. The study is conducted considering the full version of the proposed model, where the matrix {\bf B}  (see Equation (\ref{BCP_equation})) is non-diagonal. The vector of parameter in this case is $\boldsymbol{\theta} = (\alpha_{1}, \alpha_{2}, \beta_{11}, \beta_{12},\beta_{21}, \beta_{22},\omega_1, \omega_2, \phi)^\top$ and the true values used to generate the data are the following: Configuration \textbf{(a)}  $\boldsymbol{\theta}=(0.3, 0.2, 0.3, 0.1, 0.2, 0.2, 1.0, 1.0, 0.1)^\top$ and Configuration  \textbf{(b)}  $\boldsymbol{\theta}=(0.3, 0.2, 0.3, 0.1, 0.2, 0.2, 1.0, 1.0, -0.1)^\top$. 

\begin{table}
	\centering
	\begin{tabular}{@{}llcccccccccc@{}}
		\toprule
		& \multicolumn{1}{c}{$n$}              & \textbf{}                & $\alpha_{1}$  & $\alpha_{2}$ & $\beta_{11}$ & $\beta_{12}$ & $\beta_{21}$ & $\beta_{22}$ & $\omega_{1}$ & $\omega_{2}$ & $\phi $   \\ \midrule
		\multirow{6}{*}{\textbf{(a)}}          & \multirow{3}{*}{200}  & Mean     &    0.278 & 0.195 & 0.287 & 0.100 & 0.201 & 0.187 & 1.105 & 1.042 & 0.100 \\
		
		&                                   
		& SD  & 0.161 & 0.158 & 0.077 & 0.069 & 0.071 & 0.077 & 0.410 & 0.361 & 0.023  \\
		&                                   
		& MSE     &  0.026 & 0.025 & 0.006 & 0.005 & 0.005 & 0.006 & 0.179 & 0.132 & 0.001  \\ \cmidrule(l){3-12}

		& \multirow{3}{*}{500} & 
		Mean   &  0.287 & 0.197 & 0.297 & 0.102 & 0.202 & 0.194 & 1.044 & 1.020 & 0.100   \\
		&                                   & 
		SD      & 0.105 & 0.116 & 0.048 & 0.045 & 0.048 & 0.048 & 0.263 & 0.261 & 0.014  \\
		&                                   & 
		MSE     &  0.011 & 0.013 & 0.002 & 0.002 & 0.002 & 0.002 & 0.071 & 0.068 & 0.000    \\
		
		\midrule 
		
		\multirow{6}{*}{\textbf{(b)}} & \multirow{3}{*}{200}  & 
		
		Mean   & 0.266 & 0.215 & 0.292 & 0.104 & 0.198 & 0.181 & 1.123 & 1.012 & $-$0.100    \\
		&                                   & 
		
		SD  & 0.172 & 0.184 & 0.073 & 0.067 & 0.066 & 0.073 & 0.449 & 0.440 & 0.025  \\
		&                                   & 
		
		MSE   &  0.031 & 0.034 & 0.005 & 0.004 & 0.004 & 0.006 & 0.216 & 0.193 & 0.001   \\  \cmidrule(l){3-12}

		& \multirow{3}{*}{500} & 
		Mean   & 0.282 & 0.194 & 0.297 & 0.100 & 0.202 & 0.199 & 1.062 & 1.015 & $-$0.100    \\
		&                                   & 
		SD      & 0.119 & 0.138 & 0.047 & 0.045 & 0.044 & 0.045 & 0.295 & 0.309 & 0.016 \\
		&                                   & 
		MSE     &0.015 & 0.019 & 0.002 & 0.002 & 0.002 & 0.002 & 0.091 & 0.095 & 0.000  \\
		\midrule \vspace{0.3cm}
	\end{tabular} 
	\caption{Empirical mean, standard errors (SD), and mean squared error (MSE) of the Monte Carlo estimates for the BCP-INGARCH model under Configurations \textbf{(a)} and \textbf{(b)}.}\label{scenario2}
\end{table}

We simulate data from the BCP-INGARCH process with a burn-in period of 300 iterations to reduce the influence of the initial values in the simulated series. Results are reported in Table \ref{scenario2}, which displays the empirical mean and standard errors (SD) of the parameters as well as the mean squared error (MSE). This shows that the conditional maximum likelihood estimation behaves well, displaying average estimates that are close to the true values used to generate the data. In particular, the correlation parameter $\phi$ is very well estimated under both sample sizes and configurations. 
As expected, the standard deviations and MSEs decrease as the sample size increases.

We conclude that the conditional maximum likelihood estimation works well for all parameters of the BCP-INGARCH(1,1) process, producing an excellent average estimation even for the small sample size of $n=200$. Increasing the sample size produces estimates with low variability, as reflected by the decrease in standard deviation and mean squared errors.

\subsection{Standard errors}\label{ses}

Alternatives proposed in the literature to obtain the standard errors of INGARCH model parameters are investigated in this section. We consider the following consistent estimators for the Fisher information matrix \citep{ferland2006}:
\begin{eqnarray*}
	\widehat{\bm S}_n = \frac{1}{n}\sum_{t=2}^{n} \bm U_t(\widehat{\boldsymbol{\theta}}) \bm U_t^T(\widehat{\boldsymbol{\theta}})\quad\mbox{and}\quad
	\widehat{\bm D}_n = -\frac{1}{n} \sum_{t=2}^{n} \bm H_t(\widehat{\boldsymbol{\theta}}),
\end{eqnarray*}
where $\bm U_t(\cdot)$ and $\bm H_t(\cdot)$ denote the score function and Hessian matrix associated to the $t$-th bivariate count observation, respectively, for $t=2,\ldots,n$, and $\widehat{\boldsymbol{\theta}}$ is the CML estimator of $\boldsymbol{\theta}$. We employ both these methods to obtain the standard errors of the BCP-INGARCH model parameter estimates. Additionally, a parametric bootstrap alternative is evaluated, which is carried as follows. For a particular data set, model parameters are estimated via conditional maximum likelihood and used to simulate 500 trajectories of the process. The model is fitted to each replica and the bootstrap standard errors are given by the empirical standard deviations of the 500 bootstrap estimates.

A simulation study is carried to evaluate the performance of the three approaches in contrast. For simplicity, parameter settings are set with diagonal {\bf B} matrix and then we consider the true parameter vector $\boldsymbol{\theta} = (\alpha_{1}, \alpha_{2}, \beta_{11}, \beta_{22}, \omega_1, \omega_2, \phi)^\top= (0.4, 0.3, 0.2, 0.4, 1, 0.5, 0.7)^\top$. The study is repeated for the sample sizes of 100 and 500 to evaluate the methods under a small and moderate sample size. Table \ref{se_study_1} contains the results due to $n=100$, where for each model parameter we report the empirical mean, standard deviation, and median of the standard errors based on $\widehat{\bm S}_n$, $\widehat{\bm D}_n$, and bootstrap. Ideally, we would like these to approach the Monte Carlo standard deviation, that is given in bold.

\begin{sidewaystable}
	\begin{tabular}{llllllllll}
		\midrule
		& \textbf{Method}   & \textbf{Mean} & \textbf{SD} & \textbf{Median} &                                                                                         & \textbf{Method}   & \textbf{Mean} & \textbf{SD} & \textbf{Median} \\ \midrule
		\multirow{3}{*}{\begin{tabular}[c]{@{}l@{}}$\alpha_{1}$\\ \textbf{MC: 0.208}\end{tabular}} & $\widehat{\bm S}_n$     & 0.307         & 0.161        & 0.272           & \multirow{3}{*}{\begin{tabular}[c]{@{}l@{}}$\omega_{1}$\\ \textbf{MC: 0.479}\end{tabular}} & $\widehat{\bm S}_n$     & 0.693         &    0.396    & 0.602           \\
		& $\widehat{\bm D}_n$     & 0.254         & 0.187       & 0.207           &                                                                                         & $\widehat{\bm D}_n$     & 0.576         & 0.410       & 0.471           \\
		& Boot. & 0.205         & 0.044       & 0.200           &                                                                                         & Boot.  & 0.492         & 0.109       & 0.478           \\ \midrule
		\multirow{3}{*}{\begin{tabular}[c]{@{}l@{}}$\alpha_{2}$\\ \textbf{MC: 0.184}\end{tabular}} & $\widehat{\bm S}_n$     & 0.265         & 0.155       & 0.224          & \multirow{3}{*}{\begin{tabular}[c]{@{}l@{}}$\omega_{2}$\\ \textbf{MC: 0.229}\end{tabular}} & $\widehat{\bm S}_n$     & 0.309         & 0.162     & 0.273           \\
		& $\widehat{\bm D}_n$    & 0.206         & 0.170       & 0.166           &                                                                                         & $\widehat{\bm D}_n$     & 0.250         & 0.165       & 0.215           \\
		& Boot. & 0.186         & 0.044       & 0.200           &                                                                                         & Boot. & 0.259         & 0.194       & 0.242           \\  \midrule
		\multirow{3}{*}{\begin{tabular}[c]{@{}l@{}}$\beta_{11}$\\ \textbf{MC: 0.061}\end{tabular}}  & $\widehat{\bm S}_n$     & 0.076         & 0.015       & 0.074           & \multirow{3}{*}{\begin{tabular}[c]{@{}l@{}}$\phi$\\ \textbf{MC: 0.056}\end{tabular}}       & $\widehat{\bm S}_n$    & 0.076         & 0.024       & 0.074           \\
		& $\widehat{\bm D}_n$    & 0.065         & 0.028       & 0.063           &                                                                                         & $\widehat{\bm D}_n$     & 0.056         & 0.019       & 0.055           \\
		& Boot. & 0.063         & 0.006       & 0.063           &                                                                                         & Boot.  & 0.061         & 0.011       & 0.060           \\  \midrule
		\multirow{3}{*}{\begin{tabular}[c]{@{}l@{}}$\beta_{22}$\\ \textbf{MC: 0.158}\end{tabular}}  & $\widehat{\bm S}_n$     & 0.205         & 0.015       & 0.193           &                                                                                         &                   &               &             &                 \\
		& $\widehat{\bm D}_n$     & 0.172         & 0.086       & 0.160           &                                                                                         &                   &               &             &                 \\  
		& Boot. & 0.156         & 0.025       & 0.157           &                                                                                         &                   &               &             &              \\ \midrule   
	\end{tabular}
	\caption{Results of Monte Carlo simulation study for the standard errors of the BCP-INGARCH model parameters with $n=100$.}\label{se_study_1}
\end{sidewaystable}

Under a small sample size, it is observed that the best results are due to the bootstrap approach. Its mean estimates are close to the Monte Carlo standard deviation of all model parameters, something that is not true for the other methods under consideration. Moreover, the standard errors of bootstrap estimates are the smallest in all cases except for $\omega_{2}$. This method is especially advantageous concerning the $\alpha$'s parameters where $\widehat{\bm S}_n$ and $\widehat{\bm D}_n$ produce asymmetrically distributed estimates with large standard errors. 

Sample size is increased to $n=500$ and results are given in Table \ref{se_study_2}. As expected, the asymptotic methods show improved performance with the increase in sample size. The three methods now behave well for all model parameters, approximating closely the Monte Carlo standard deviations, with a slight disadvantage in the $\widehat{\bm S}_n$ approach. As before, the standard deviations due to the bootstrap alternative are the smallest among the three methods, but the difference is now subtle.

\begin{sidewaystable}
	\begin{tabular}{llllllllll}
		\midrule
		& \textbf{Method}   & \textbf{Mean} & \textbf{SD} & \textbf{Median} &                                                                                         & \textbf{Method}   & \textbf{Mean} & \textbf{SD} & \textbf{Median} \\ \midrule
		\multirow{3}{*}{\begin{tabular}[c]{@{}l@{}}$\alpha_{1}$ \\ \textbf{MC: 0.096}\end{tabular}} & $\widehat{\bm S}_n$     & 0.101         & 0.021        & 0.099           & \multirow{3}{*}{\begin{tabular}[c]{@{}l@{}}$\omega_{1}$\\  \textbf{MC: 0.209}\end{tabular}} & $\widehat{\bm S}_n$     & 0.220         &    0.050    & 0.214           \\
		& $\widehat{\bm D}_n$     & 0.094         & 0.025       & 0.090           &                                                                                         & $\widehat{\bm D}_n$     & 0.206         & 0.056       & 0.197           \\
		& Boot. & 0.099         & 0.017       & 0.096           &                                                                                         & Boot.  & 0.218         & 0.039       & 0.214           \\ \midrule
		\multirow{3}{*}{\begin{tabular}[c]{@{}l@{}}$\alpha_{2}$\\  \textbf{MC: 0.079}\end{tabular}} & $\widehat{\bm S}_n$     & 0.081         & 0.020       & 0.080          & \multirow{3}{*}{\begin{tabular}[c]{@{}l@{}}$\omega_{2}$\\  \textbf{MC: 0.092}\end{tabular}} & $\widehat{\bm S}_n$     & 0.094         & 0.019     & 0.092           \\
		& $\widehat{\bm D}_n$     & 0.075         & 0.022       & 0.072           &                                                                                         & $\widehat{\bm D}_n$     & 0.088         & 0.020       & 0.085           \\
		& Boot. & 0.079         & 0.013       & 0.077           &                                                                                         & Boot. & 0.092         & 0.014       & 0.091           \\  \midrule
		\multirow{3}{*}{\begin{tabular}[c]{@{}l@{}}$\beta_{11}$\\  \textbf{MC: 0.027}\end{tabular}}  & $\widehat{\bm S}_n$     & 0.029         & 0.003       & 0.029           & \multirow{3}{*}{\begin{tabular}[c]{@{}l@{}}$\phi$\\  \textbf{MC: 0.019}\end{tabular}}       & $\widehat{\bm S}_n$     & 0.023         & 0.005       & 0.023           \\
		& $\widehat{\bm D}_n$     & 0.027         & 0.004       & 0.027           &                                                                                         & $\widehat{\bm D}_n$     & 0.020         & 0.005       & 0.020           \\
		& Bootstrap & 0.027         & 0.002       & 0.027           &                                                                                         & Boot.  & 0.020         & 0.002       & 0.020           \\  \midrule
		\multirow{3}{*}{\begin{tabular}[c]{@{}l@{}}$\beta_{22}$\\  \textbf{MC: 0.069}\end{tabular}}  & $\widehat{\bm S}_n$     & 0.073         & 0.011       & 0.072           &                                                                                         &                   &               &             &                 \\
		& $\widehat{\bm D}_n$     & 0.069         & 0.009       & 0.068           &                                                                                         &                   &               &             &                 \\  
		& Boot. & 0.069         & 0.007       & 0.069           &                                                                                         &                   &               &             &              \\   \midrule 
	\end{tabular}
	\caption{Results of Monte Carlo simulation study for the standard errors of the BCP-INGARCH model parameters with $n=500$.}\label{se_study_2}
\end{sidewaystable}

We conclude with the recommendation of using the bootstrap approach for obtaining the standard errors of the BCP-INGARCH model estimates when the sample size is small. In this situation, the asymptotic based methods tend to overestimate this quantity and also have larger variability. When the sample size is moderate as 500 observations, the $\widehat{\bm D}_n$ strategy behaves well and can be considered to minimize computational cost.

\subsection{Hypothesis testing}\label{hyp_test}

Methods to test for the presence of contemporaneous correlation between the count time series are evaluated in this section. Under a BCP-INGARCH model we would like to test the hypothesis $H_0: \phi = 0 $ versus $H_1: \phi \neq 0$. We evaluate two asymptotic tests, the likelihood ratio test (LRT) and the score test via a simulation study. The Monte Carlo probability of rejecting $H_0$ in favor of $H_1$ is calculated with 1000 replicas for each $\phi$ in the range $[-1, 1]$ with 0.1 spacing. In each Monte Carlo replica, model parameters are estimated under $H_0$ and $H_1$, denoted as $\widetilde{\bm\theta}$ and $\widehat{\bm\theta}$, respectively. The likelihood ratio test statistic is calculated as $-2(\ell(\widetilde{\bm\theta}) - \ell(\widehat{\bm\theta}))$, where $\ell(\cdot)$ is the log-likelihood function. The score test relies only on the model parameters estimated under $H_0$ and its test statistic is given by $\bm U(\widetilde{\theta})^\top \bm I^{-1}(\widetilde{\theta}) \bm U(\widetilde{\theta})$, where $\bm U(\cdot)$ and $\bm I(\cdot)$ denotes the score function and the model's Fisher Information matrix, respectively, with the first being calculated analytically via expressions provided in Subsection \ref{mle} and the former by numerical differentiation.

The simulation study is carried for a setting where {\bf A} and {\bf B} are diagonal matrices with the true parameter values $(\alpha_{1}, \alpha_2, \beta_{11}, \beta_{22}, \omega_{1}, \omega_{2})=(0.4, 0.3, 0.2, 0.4, 1, 1)^\top$. We refer to this specification as Scenario I.  A configuration where the matrix {\bf B} is non-diagonal is chosen in Scenario II, with the parameter vector $(\alpha_{1}, \alpha_2, \beta_{11}, \beta_{12}, \beta_{21}, \beta_{22}, \omega_{1}, \omega_{2})=(0.3, 0.2, 0.3,0.1, 0.2, 0.2, 1, 0.5)$. We set the significance level of both tests at 5\%. Figure \ref{LRT_Score} displays the power of the likelihood ratio and score tests as function of $\phi$ under both Scenarios 1 and 2 and sample sizes equal to 100 and 500.

\begin{figure}
	\centering
	\includegraphics[width=0.7\textwidth]{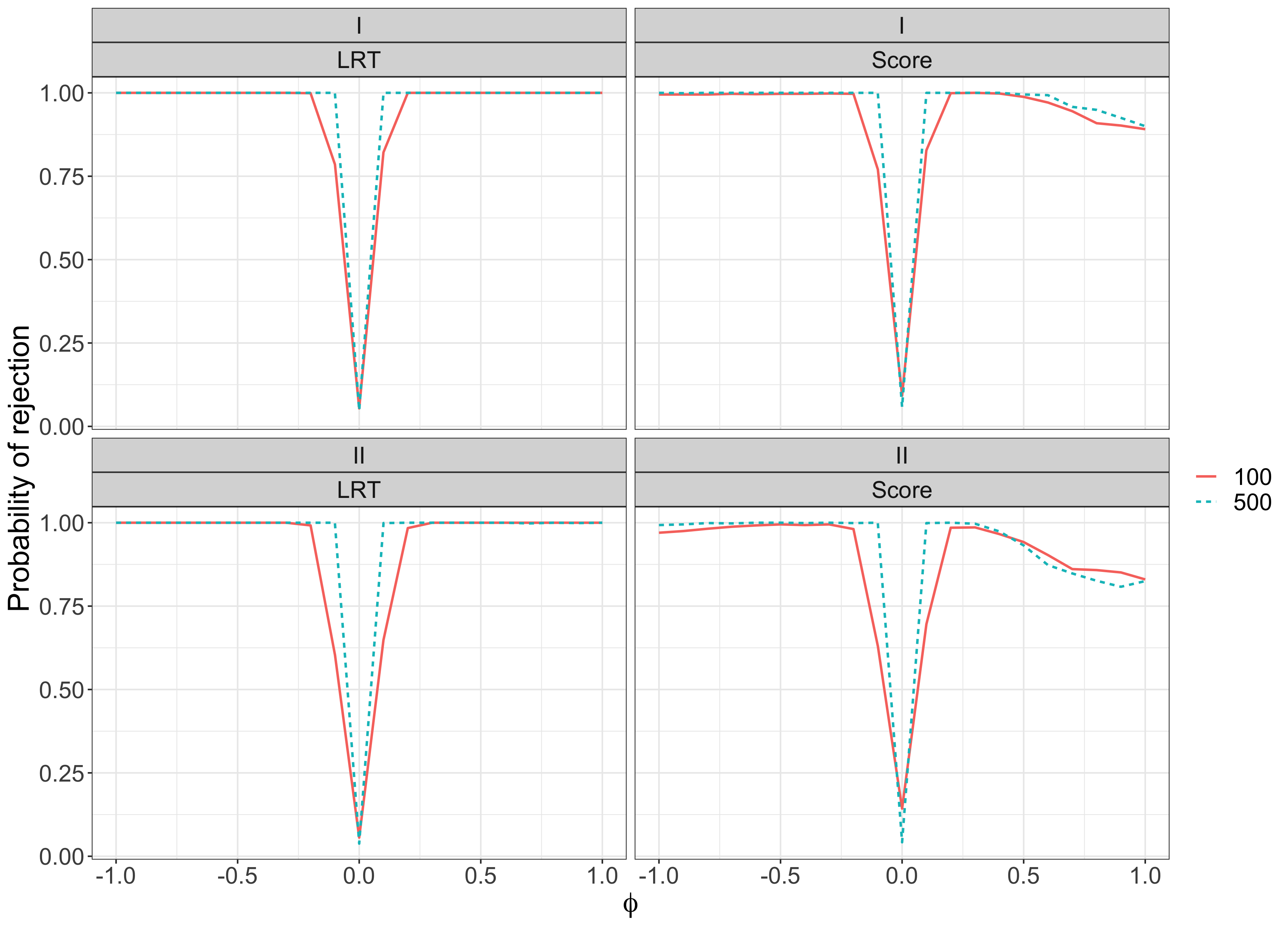}
	\caption{Power of the likelihood ratio test (LRT) and score test as function of $\phi$ with data generated under Scenarios I and II and sample sizes $n=100$ (solid line) and $n=500$ (dashed line).}\label{LRT_Score}
\end{figure}

Both the likelihood ratio and score tests demonstrate the ability to reject the null hypothesis (power) with a high probability when $\phi\neq0$. However, the score test suffers from numerical problems at high positive values of $\phi$ as can be observed from Figure \ref{LRT_Score}. This issue arises from numeric differentiation employed to calculate the Hessian matrix, causing the rejection probability to decrease in this region. Notably, this is more severe in Scenario II where the number of parameters increases.


\section{Bivariate hepatitis count data analysis}\label{application}

\subsection{BCP-INGARCH modeling}

In this section, the proposed methodology is applied to the confirmed monthly cases of viral hepatitis recorded at two nearby Brazilian cities. Hepatitis is an inflammation of the liver, most commonly caused by a viral infection. Symptoms can take some time to develop, only manifesting after the liver function has been affected. In Brazil, the most common types of hepatitis are A, B, and C. The data is made available by the Brazilian public healthcare system SUS through the site \url{https://datasus.saude.gov.br} (DATASUS platform). It currently comprises the period of 2001 to 2018, giving a total of $n=216$ observations per city. We analyze the data of Brazil's capital Bras\'ilia, which is located in the Federal District within Goi\'as state. Due to close proximity, it is natural to expect that Bras\'ilia's counts are correlated to the Goi\'as's capital, Goi\^ania. The goal here is the joint modeling and prediction of the monthly counts of hepatitis cases in Bras\'ilia and Goi\^ania. The empirical Pearson's correlation between series is 0.50. As we will discuss in Subsection \ref{comparison_app}, the existing bivariate INGARCH models cannot handle this problem due to the constrained parameter space of $\phi$ (see also discussion in Subsection \ref{other_ingarch}).

The geographic locations of Bras\'ilia and Goi\^ania are shown on the right side of Figure \ref{brasilia_goiania}. Time series of viral hepatitis confirmed cases in each city are displayed on the left side of Figure \ref{brasilia_goiania}, showing that the pair tends to be correlated over time. One of the benefits of modeling the cross-correlation is that the current data in one city can be used to predict the future of another city. This situation will be illustrated in Subsection \ref{cond_prediction}.

\begin{figure}
	\centering
	\includegraphics[width=0.85\linewidth]{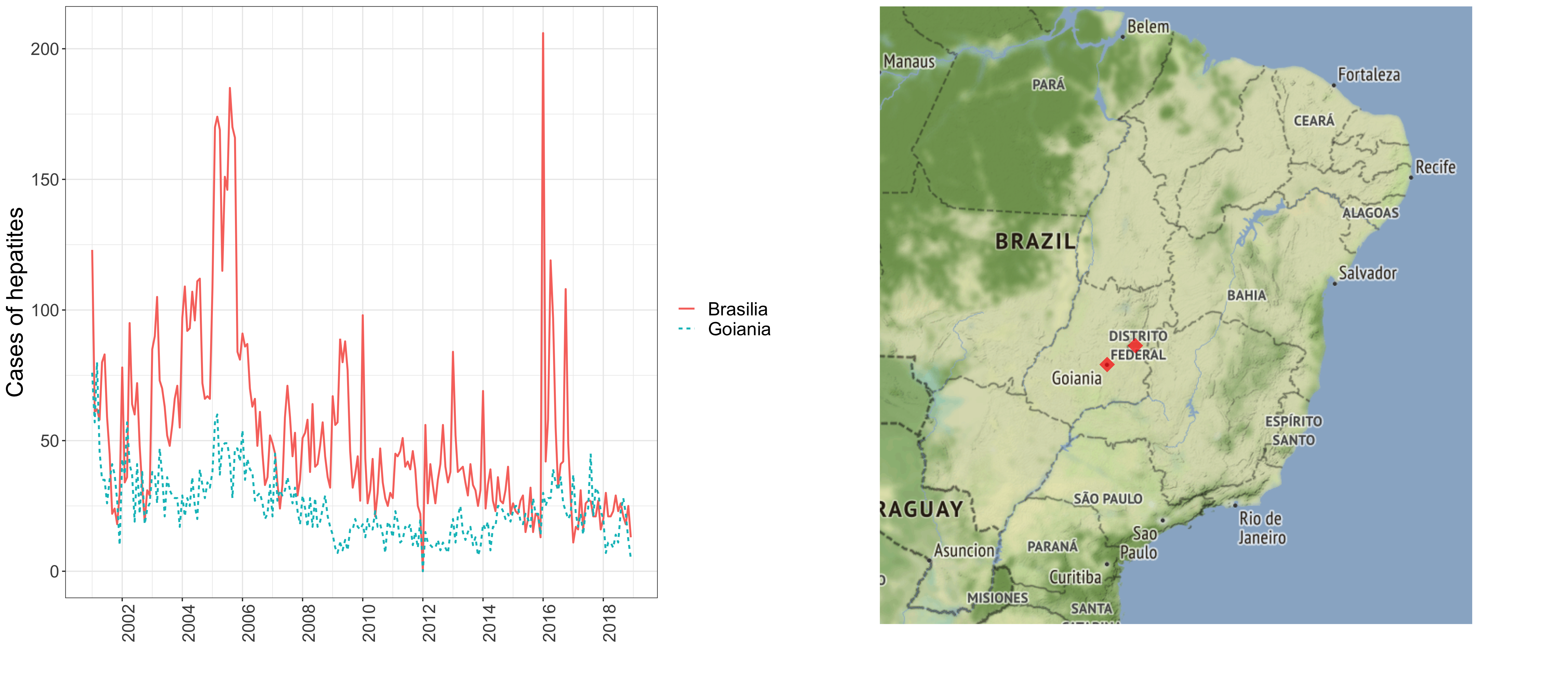} 
	\caption{On the left, monthly counts of viral hepatitis confirmed cases in the Brazilian cities of Bras\'ilia and Goi\^ania from 2001 to 2018. Empirical Pearson's correlation between series is 0.50. On the right, geographical representation of Brazilian cities Bras\'ilia and Goi\^ania.}\label{brasilia_goiania}
\end{figure}

The lagged relationship among the series is explored in Figure \ref{lead_lag}, where we assess how the counts of hepatitis of one city at time $t-1$ are related to those of the other city at time $t$. The scatter plots suggest a positive correlation between Goi\^ania$(t-1)$ and Bras\'ilia$(t)$ as well as among Goi\^ania$(t)$ and Bras\'ilia$(t-1)$. This conveys that the non-diagonal BCP-INGARCH fit should be considered in our analysis. In this case, the matrix {\bf B} permits non-zero $\beta_{12}$ and $\beta_{21}$ coefficients which means that lagged counts of one series may affect the future marginal mean of another (see Equation (\ref{BCP_equation})).

\begin{figure}
	\centering
	\includegraphics[width=0.7\textwidth]{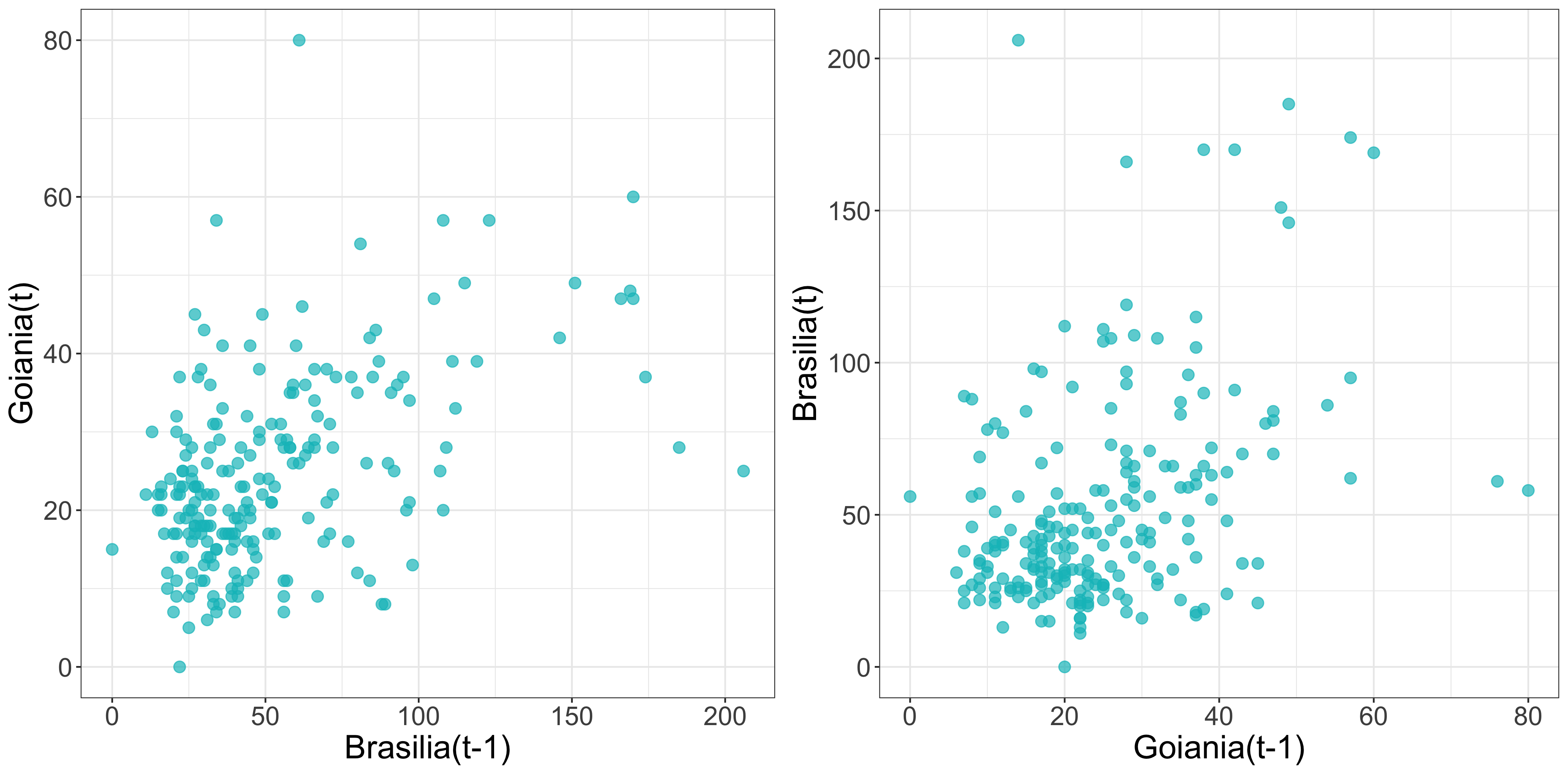}
	\caption{Lagged relationship between monthly counts of hepatitis at Bras\'ilia and Goi\^ania. On the left, Goi\^ania counts at time $t$ are plotted versus Bras\'ilia counts at $t-1$. Similarly on the right, with Bras\'ilia as the leading component.}\label{lead_lag}
\end{figure}

To fit the BCP-INGARCH model a data-driven approach is taken for choosing how to assign the observed count time series to $\{Y_{1t}\}$ and $\{Y_{2t}\}$ (remember that $Y_{1t}|\mathcal F_{t-1}$ and $Y_{2t}|\mathcal F_{t-1}$ are marginally Poisson and mixed Poisson distributed). Since the second component $Y_{2t}$ has more sources of overdispersion than $Y_{1t}$, the empirical dispersion index $D$ (empirical variance divided by the empirical mean) is calculated for the two series, and that with a higher value of $D$ is assigned to $Y_{2t}$. We obtain $D =  23.89$ for Bras\'ilia counts and $D = 6.45$ for Goi\^ania, hence the pair $Y_{1t}$ and $Y_{2t}$ denote the monthly count of hepatitis cases of Goi\^ania and Bras\'ilia, respectively, for $t=1,\ldots,216$.

Diagonal (fit 1) and non-diagonal (fit 2) BCP-INGARCH processes are fitted to the data. The parameter estimates and standard errors of the BCP-INGARCH model parameters are reported in Table \ref{model_fit}. Standard errors are obtained via parametric bootstrap with 500 replicas since this method is more reliable for small sample sizes. Both fits indicate a positive serial auto-correlation in both count time series and a high cross-correlation is evidenced by the ``small" estimate value of $\widehat{\phi}$. Moreover, the estimated contemporaneous correlation (conditional on the past) can be obtained by calculating the BCP correlation (Equation (\ref{correlation})) for every $t$ with the CML parameter estimates. This is displayed in Figure \ref{cross_corr_plot}, where a decreasing tendency of cross-correlation is shown, although peaks occur around 2005 and 2016. These are periods with pronounced peaks of cases in Bras\'ilia, associated with an increase in Goi\^ania. The conditional (on the past) contemporaneous-correlation estimated by the diagonal and non-diagonal BCP-INGARCH model fits are very similar and mostly overlap.

\begin{figure}
	\centering
	\includegraphics[width=0.7\textwidth]{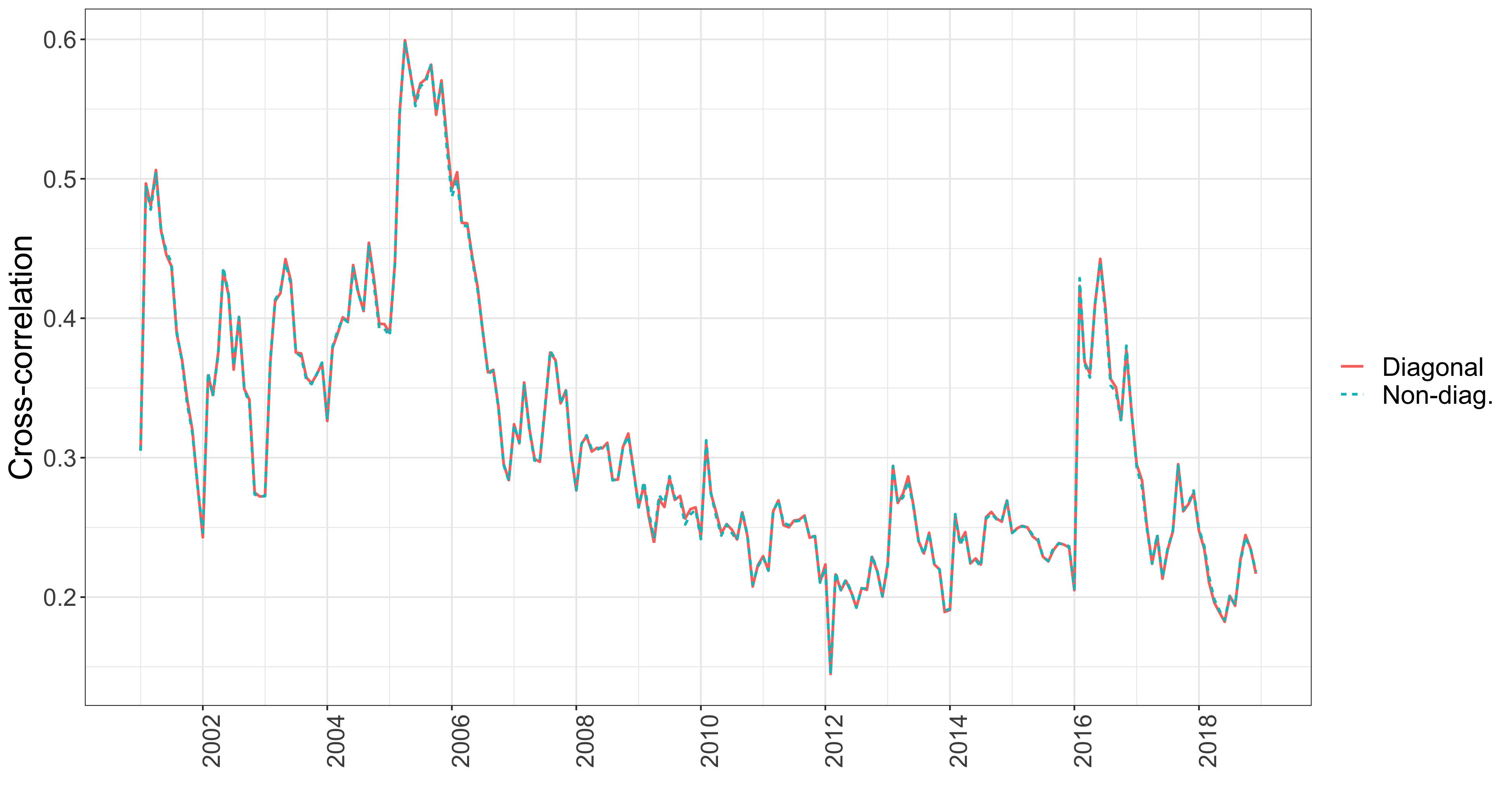}
	\caption{Dynamics of the conditional contemporaneous correlation between count time series of confirmed hepatitis cases at Bras\'ilia and Goi\^ania estimated by the BCP-INGARCH(1,1) model fits.}\label{cross_corr_plot}
\end{figure}

We test for the presence of cross-correlation ($H_0: \phi=0$ versus $H_1: \phi \neq 0$) through the likelihood ratio and score tests discussed in Section \ref{hyp_test} for the diagonal and non-diagonal BCP-INGARCH models. The likelihood ratio test produces test statistics (and $p$-values in parenthesis) of 68.06 ($1 \times 10^{-16}$) and 61.02 ($5 \times 10^{-15}$) for the diagonal and non-diagonal fits, respectively. Those due to the score test are 69.13 ($1\times 10^{-16}$) and 66.48 ($3\times 10^{-16}$). Both tests strongly reject the null hypothesis, so we can conclude that there is a statistically significant contemporaneous correlation between the count time series.

This high cross-correlation is expected due to the close proximity of the cities and relates directly to the forms of transmission of viral hepatitis. The main kinds of transmission of hepatitis type B and C are sexual contact and parenteral form, respectively, while for hepatitis A contagion is a predominantly fecal-oral route. Hence, one possible explanation for a high cross-correlation among the confirmed cases of the disease at Bras\'ilia and Goi\^ania is the sharing of food and water sources, that relate directly to the transmission of hepatitis A. From an epidemiological perspective, it would be interesting to repeat this study considering separately the hepatitis cases of each type and exploring whether the contemporaneous behavior maintains as strong when modeling the three different classes independently.

\begin{table}[]
	\centering
	\begin{tabular}{@{}lccccc@{}}
		\toprule
		& \textbf{$\alpha_{11}$}            & \textbf{$\alpha_{22}$}            & \textbf{$\beta_{11}$}             & \textbf{$\beta_{12}$}             & \textbf{$\beta_{21}$}             \\ \midrule
		Fit 1 & 0.466 (0.140)                     & 0.482 (0.147)                     & 0.430 (0.074)                     & $-$ ($-$)                                 & $-$ ($-$)                                 \\
		Fit 2 & \multicolumn{1}{l}{0.497 (0.134)} & \multicolumn{1}{l}{0.448(0.137)}  & \multicolumn{1}{l}{0.398 (0.061)} & \multicolumn{1}{l}{0.002 (0.017)} & \multicolumn{1}{l}{0.019 (0.039)} \\ \midrule
		& \textbf{$\beta_{22}$}             & \textbf{$\omega_{1}$}             & \textbf{$\omega_{2}$}             & \textbf{$\phi$}                   &                                   \\ \midrule
		Fit 1 & \multicolumn{1}{l}{0.384 (0.066)} & \multicolumn{1}{l}{2.310 (1.948)} & \multicolumn{1}{l}{6.519 (5.419)} & \multicolumn{1}{l}{0.010 (0.002)} & \multicolumn{1}{l}{}              \\
		Fit 2 & 0.399 (0.065)                     & 2.216 (1.881)                     & 7.012 (4.814)                     & 0.010 (0.002)                     &                                   \\ \bottomrule \vspace{0.3cm}
	\end{tabular} 
	\caption{Parameter estimates and standard errors (under parenthesis) of the BCP-INGARCH(1,1) model fits to the hepatitis data in the Brazilian capitals of Goi\^ania($y_1$) and Bras\'ilia($y_2$). Fits 1 and 2 correspond to different setting of {\bf B} matrix as diagonal or non-diagonal, respectively. Standard errors are obtained from parametric bootstrap with 500 replicas.}\label{model_fit}
\end{table}

One way to select between the diagonal and non-diagonal alternatives is through information criteria. We consider the AIC and BIC as selection criteria and both of them indicate that the diagonal model is preferred over the non-diagonal model. The model information criteria suggest that a leading/lagging relationship among the counts is not statistically significant. In other words, the effects of lagged Bras\'ilia counts in Goi\^ania and vice-versa are not statistically different than zero. This implies that cross-correlation is purely contemporaneous, and not driven by a leading/lagging relationship among the count time series.

\subsection{Out-of-sample prediction}\label{oos}

Although the diagonal fit is preferable from a model selection perspective, the conclusion can be different if the main focus of the practitioner is on prediction. In this subsection, we evaluate the out-of-sample prediction due to the diagonal and non-diagonal BCP-INGARCH processes. More specifically, we consider the one-step-ahead forecast performed 100 times recursively. This mimics the situation where data is collected from both cities and, at each month, used to predict the outcome of the following. 

A point prediction for $\widehat{\bf Y}_{t+1}$ is given by the joint mode of the distribution $\mbox{BCP}(\widehat{\lambda}_{t,1}, \widehat{\lambda}_{t,2}, \widehat{\phi})$, with $\widehat{\boldsymbol{\lambda}}$ estimated using the CMLEs. We start fitting the model with the observed values until August 2010 $(t = 116)$ (${\bf Y}_t$), which is used to predict September 2010 $(t = 117)$ (${\widehat{\bf Y}_{t+1}}$). Once completed, the model is refitted incorporating the true data up to $t = 117$. This is carried until the last observed data point, resulting in a 100 out-of-sample one-step-ahead set of predictions for Goi\^ania and Bras\'ilia. We evaluate the results through the root mean-square forecasting error (RMSFE) loss function
$\mbox{RMSFE}_{t\,i} =  \sqrt{ \frac{1}{t-n_0} \sum_{s = n_0 +1}^{t} (Y_{si} - \widehat{Y}_{si})^2 }$, for $i=1,2$ and $t=n_0 + 1,\ldots, 216$, with $n_0 = 116$. 

\begin{figure}
\centering
\includegraphics[width=0.8\textwidth]{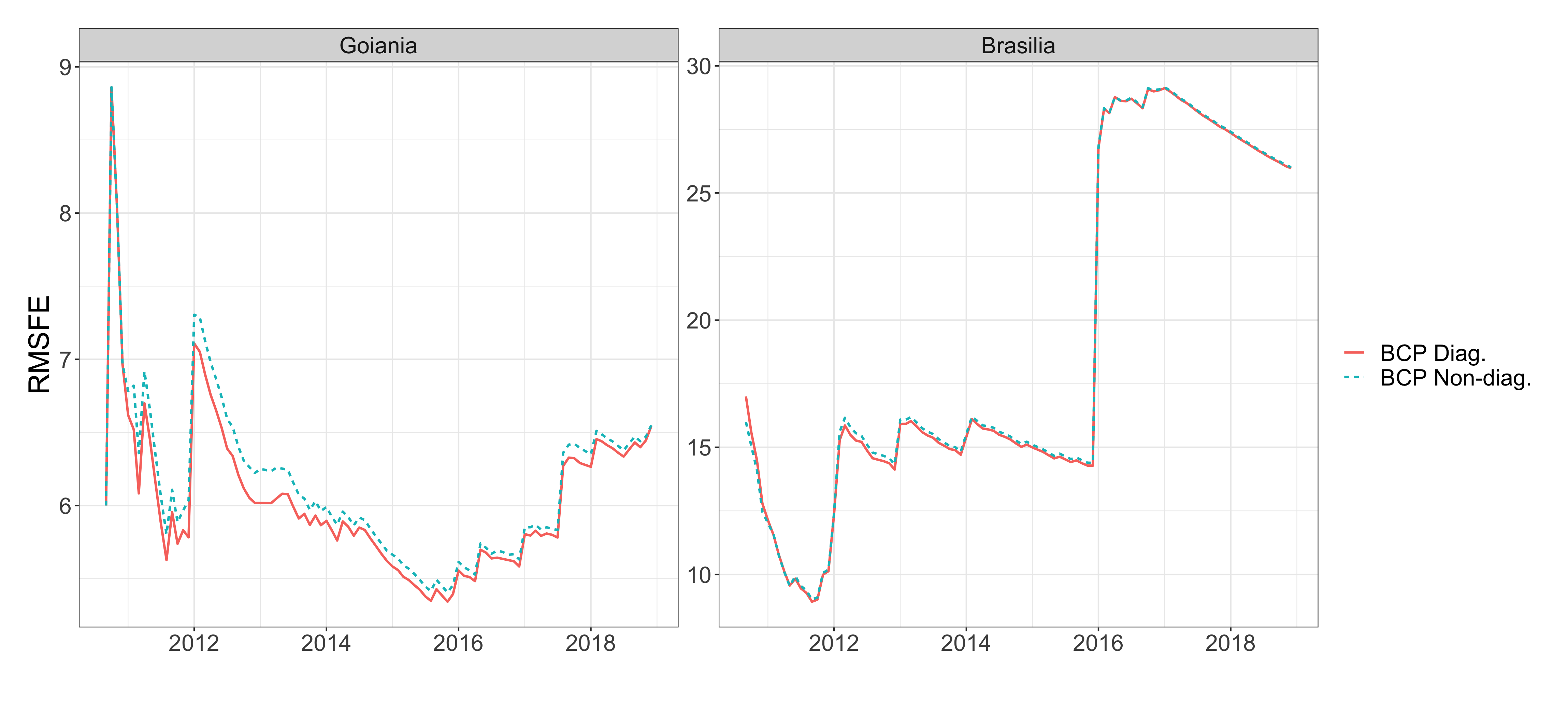}
\caption{Root mean-square forecasting error (RMSFE) versus time for Goi\^ania and Bras\'ilia obtained with the diagonal (solid line) and non-diagonal (dashed line) BCP-INGARCH processes.}\label{msfe_joint}
\end{figure}

The results given in Figure \ref{msfe_joint} show that the RMSFE due to the diagonal and non-diagonal BCP-INGARCH(1,1) processes are similar for Bras\'ilia under the entire forecasting period but the preferred model for Goi\^ania is consistently the diagonal option. Table \ref{pred_table} provides a numerical quantification of the overall forecasting period via the root-mean-squared error (RMSE) and mean absolute error (MAE) of the 100 one-step-ahead predictions.

\begin{table}[]
	\centering
\begin{tabular}{@{}clcc@{}}
	\toprule
	\textbf{City}                                & \multicolumn{1}{c}{\textbf{Model}} &\textbf{RMSE}   & \textbf{MAE} \\ \midrule
	\multirow{3}{*}{Bras\'ilia}                    & Diag. BCP-INGARCH                  & 25.975 & 14.670       \\
	& Non-diag. BCP-INGARCH              & 26.016 & 14.770       \\ \midrule
	\multicolumn{1}{l}{\multirow{3}{*}{Goi\^ania}} & Diag. BCP-INGARCH                  & 6.543  & 5.070        \\
	\multicolumn{1}{l}{}                         & Non-diag. BCP-INGARCH              & 6.549  & 5.110        \\ \midrule
\end{tabular}
\caption{RMSE and MAE of 100 one-step-ahead predictions of confirmed hepatitis cases in Bras\'ilia and Goi\^ania from BCP-INGARCH(1,1) processes. }\label{pred_table}
\end{table}

For both cities, the BCP-INGARCH process resulting in the lowest forecasting error was the diagonal option, which is in accordance with the model selection indicated previously.  A proper prediction of hepatitis cases as provided by our model is crucial for the planning of public health resources, where an accurate forecast will help guide the allocation of public resources to medication purchases, hospital beds, and vaccination, for example.

\subsection{Conditional prediction}\label{cond_prediction}

Another advantage of the proposed model is that it can be used to perform conditional prediction for one of the time series given that the value of the other is known. An example of an application where this is of interest is if one of Goi\^ania or Bras\'ilia is more efficient in reporting their number of confirmed hepatitis cases. We can incorporate this information to predict the value for the other city where reporting was late. 

Here we illustrate how this extra piece of information can help improving prediction for the city of interest in contrast to jointly forecasting the pair. A similar out-of-sample one-step-ahead prediction as in Subsection \ref{oos} is used for this purpose and repeated 100 times as before. The difference now is that we assume that the city we have extra information on is assigned to $\{Y_{1t}\}$ (Goi\^ania), and the mode of a Poisson distribution with parameter $\widehat{\lambda}_{2t} \exp\{-\widehat{\lambda}_{1t} (e^{\widehat{\phi}} -1)+\widehat{\phi} Y_{1\,t+1}\}$ provides the prediction of $Y_{2\,t+1}$ (this follows from the definition of a BCP distribution). Following what was done in Subsection \ref{oos}, Bras\'ilia is assigned to $\{Y_{2t}\}$ and counts of Goi\^ania are used to make the conditional prediction.

\begin{table}[]
	\centering
	\begin{tabular}{@{}clcccc@{}}
		\toprule
		\multirow{2}{*}{\textbf{City}}               & \multicolumn{1}{c}{\multirow{2}{*}{\textbf{Model}}} & \multicolumn{2}{c}{\textbf{RMSE}}                                             & \multicolumn{2}{c}{\textbf{MAE}}                                              \\ \cmidrule(l){3-6} 
		& \multicolumn{1}{c}{}                                & \multicolumn{1}{l}{\textbf{Joint}} & \multicolumn{1}{l}{\textbf{Conditional}} & \multicolumn{1}{l}{\textbf{Joint}} & \multicolumn{1}{l}{\textbf{Conditional}} \\ \cmidrule(r){1-2}
		\multirow{2}{*}{Bras\'ilia}                    & Diag. BCP-INGARCH                                   & 25.975                             & 25.876                                   & 14.67                              & 14.55                                    \\ \vspace{0.2cm}
		& Non-diag. BCP-INGARCH                               & 26.016                             & 25.901                                   & 14.77                              & 14.60                                    \\  \midrule
	\end{tabular}
\caption{RMSE and MAE of 100 one-step-ahead conditional and joint predictions of confirmed hepatitis cases in Bras\'ilia under the diagonal and non-diagonal BCP-INGARCH(1,1) processes. }\label{cond_table}
\end{table}

In Table \ref{cond_table}, we mimic the scenario where Goi\^ania is the city that discloses the number of cases first. To facilitate comparison, we include RMSE and MAE obtained from the jointly forecasting the pair $(Y_1, Y_2)$, as done in Subsection \ref{oos}. As expected, the conditional prediction yields a reduction in the prediction error in comparison to jointly forecasting the pair of counts since more information is incorporated.

\subsection{Comparison to other bivariate INGARCH models}\label{comparison_app}

We aimed to compare the results from the proposed BCP-INGARCH model to the existing bivariate INGARCH models discussed in Subsection \ref{other_ingarch}. However, we encountered problems and limitations with the models by \cite{liu2012} and \cite{cuizhu2018} that made them unfit for the application of interest in this paper, as discussed below.

Fitting the bivariate Poisson model by \cite{cuizhu2018} to the hepatitis count data of Goi\^ania and Bras\'ilia indicated that this model can be very sensitive to initial values of the cross-correlation parameter. For a careful assessment of our implementation and maximization routine, we tried to replicate the model fit to the weekly number of syphilis cases in Pennsylvania and Maryland from 2007 to 2010, reported by \cite{cuietal2019}. This data set is publically available through \texttt{R} package \texttt{ZIM} (\hyperlink{https://cran.r-project.org/web/packages/ZIM/ZIM.pdf}{https://cran.r-project.org/web/packages/ZIM/ZIM.pdf}). 

Our analysis revealed that the contribution of $\phi$ to the likelihood is small relative to the other parameters in the model, causing the estimate of $\phi$ to barely move away from its initialization. For instance, starting from $\phi_0 = (-0.9, -0.5, 0. 0.5, 0.9)$ resulted in the CMLEs $\widehat{\phi} = (-0.89999, -0.50009, 0.00068, 0.50030, 0.90002)$. The maximized log-likelihood up to proportionally for these five different initializations were $(381.31, 385.25, 387.46, 388.36, 388.52)$. If proportionally terms are considered, the difference in log-likelihood values is even more subtle, justifying the difficulty in finding the CML estimate for this parameter. 

In addition, some inconsistencies with the results reported by the authors are the following.

1. The log-likelihood up to proportionality evaluated at the CMLEs in Table 9 by \cite{cuietal2019} is 388.14. This shows that the maximization procedure adopted by the authors was unable to find the CMLEs.

2. The point estimate of $\phi$ is positive and equal to 0.7468. It is expected that $\phi$ assumes a negative value since the sign of this parameter determines the sign of the cross-correlation. The empirical Pearson's correlation for this data set is $-0.1355$.

3. It is unclear if the stationarity and ergodicity condition in Assumption \ref{RCp} was met for the CMLEs. In practice, it is common to verify such condition with $p=1$, in which case it is not satisfied for the CMLEs of ${\bf A} = (\alpha_1,\alpha_2)$ and ${\bf B} = (\beta_1, \beta_2)$ in Table 9.

In our view, the estimation of $\phi$ could be improved by performing the maximization in two steps, first estimating the parameters related to the conditional mean and $\phi$ subsequently, but this is outside of our scope.

We also tried fitting the model by \cite{liu2012} and \cite{leeetal2018} to the hepatitis data, but encountered numerical issues and were unable to achieve convergence to the CMLEs. This is likely related to the limited range of correlation of this model, where a cross-correlation as high as in this application is not supported.  As discussed in Subsection \ref{other_ingarch}, the maximum value of $\phi$ (hereby $\phi_{max}$) is given by $\min(a_1, a_2)$, where $(a_1, a_2)^T = (\boldsymbol{I} - \boldsymbol{A})^{-1} \boldsymbol{\omega}$.  We can get a rough idea of the upper bound for $\phi$ by employing the estimates from the BCP-INGARCH process. Since the INGARCH models share the same specification of the conditional mean $\boldsymbol{\lambda_t}$, it is reasonable to expect that the estimates of $\boldsymbol{A}$, $\boldsymbol{B}$ and $\boldsymbol{\omega}$ will not differ greatly. This gives us an approximate upper limit $\phi_{max} \approx 4.3$. We can investigate the cross-correlation over time achieved by the model by \cite{liu2012} and \cite{leeetal2018} by calculating the correlation due to their baseline bivariate Poisson distribution with $\phi_{max}$ and the BCP-INGARCH fitted $\boldsymbol{\widehat{\lambda}}$. This gives us that the maximum contemporaneous correlation for the study period is 0.29. Meanwhile, this is a high as 0.6 under our model, where values above 0.29 are frequently encountered, as illustrated in Figure \ref{cross_corr_plot}.

The data application explored in this paper evidenced how the existing bivariate INGARCH models can be unfit for practical problems, producing inconsistent or no results at all. This motivates further the introduction of the BCP-INGARCH process, which not only supports a broad range of positive and negative cross-correlation but also has a simple likelihood specification that helps to avoid problems that arise from numerical maximization.

\section{Concluding remarks}\label{conclusion}

We developed a novel bivariate conditional Poisson INGARCH process for modeling correlated count time series data having as the main advantage of its capability of capturing a wide range of contemporaneous correlation. This flexibility is important since bivariate/multivariate count time series data is prevalent in many fields and that there is a lack of flexible bivariate models based on the INGARCH approach, which is a relevant tool for dealing with univariate count time series. We here showed that it is possible to construct promising models based on such an approach.

The stability theory of our bivariate count process was established. Through simulation studies, we demonstrated that the proposed conditional maximum likelihood estimation works well and evaluated different methods of obtaining parameter standard errors.  The simulation studies showed that the parametric bootstrap is preferred for small sample sizes, but asymptotic alternatives work well with moderate or large samples. Asymptotic properties of the estimators were also derived. Hypothesis testing for the presence of cross-correlation under our model was presented and evaluated through likelihood ratio and score tests, demonstrating the power of such tests. Finally, the proposed methodology was employed in an application to counts of hepatitis cases at nearby Brazilian cities. The series showed to be highly positively correlated and modeling the data jointly was successfully done through our proposed model. The limitation of some bivariate INGARCH models was discussed both theoretically and empirically. 

We now discuss some possible points for future research. As demonstrated along with this paper, a key ingredient to propose a bivariate INGARCH process being mathematically tractable and having flexible contemporaneous-correlation structure relies on the baseline bivariate count distribution. For example, other models can be proposed by assuming $Z_1\sim\mbox{NB}(\lambda_1,\sigma)$ (see beginning of Section \ref{def_prop}) if an exceeding overdispersion needs to be accounted for the count time series $\{Y_{1t}\}$, where $\mbox{NB}(\lambda_1,\sigma_1)$ stands for a negative binomial distribution with mean $\lambda_1$ and dispersion parameter $\sigma_1$. In a similar fashion, by assuming $Z_2|Z_1=z_1\sim\mbox{NB}(\mu_2 e^{\phi z_1},\sigma_2)$, with $\mu_2$ as defined in Section \ref{def_prop}, we can account for a wider range of overdispersion related to $\{Y_{2t}\}$. Other alternatives to the negative binomial assumption like COM-Poisson, zero-inflated/deflated count models, or Poisson inverse-Gaussian distributions can be considered for attacking underdispersion, overdispersion, zero-inflation/deflation, and heavy-tailed counts, just to name a few. We call the attention that in any of these extensions, many of the developed methodology given in this paper can be straightforwardly adapted. It is worth to mention that an \texttt{R} package is being finished and will be available soon, which is useful for practitioners and applied statisticians as well as for comparison purposes of new emerging multivariate count time series methods. 

Another important point to be addressed is a multivariate extension allowing for a higher dimension rather than 2. The model proposed by \cite{foketal2020} allows for dealing with $d\geq2$ correlated count time series, which is an attractive feature over the existing multivariate INGARCH models (including our proposed process). We hope to address this problem in a future paper. Other topics deserving further research are (a) inclusion of covariates, which can be done via a log-linear structure like in \cite{foketal2020}; (b) $\mbox{BCP-INGARCH}(p,q)$ high-order extension with $p,q\geq1$; and (c) non-linear BCP-INGARCH generalization.

\section*{Acknowledgements}
\noindent  L.S.C. Piancastelli thanks the financial support of Science Foundation Ireland under Grant number 18/CRT/6049. W. Barreto-Souza and H. Ombao would like to acknowledge the financial support by KAUST Research Fund and NIH 1R01EB028753-01. W. Barreto-Souza also thanks to the Conselho Nacional de Desenvolvimento Cient\'ifico e Tecnol\'ogico (CNPq-Brazil, grant number 305543/2018-0).

\end{document}